\documentclass[final,conference,10pt]{IEEEtran}
\usepackage{color}
\usepackage{overpic}
\usepackage{cite}
\usepackage{amsmath,amsfonts,amssymb,amsthm,bbm,bm}

\newcommand{\pacap}{{R_{ij}^{\min}}}
\newcommand{\Cin}{{R_{ij}^{\text{in}}}}
\newcommand{\bound}{{F^{\max}}}
\newcommand{\TB}{{\text{BP-T}}}
\newcommand{\OR}{\boldsymbol\lambda{\text{--OR}}}
\newcommand{\BPa}{{\text{BP}}}
\newcommand{\BPr}{{\text{BP-O}}}
\newcommand{\BPs}{{\text{BP-SP}}}

\newcommand{\driftK}{{\Delta^{\TB}_K(t)}}

\newcommand{\mean}[1]{\mathbb{E}\!\left[#1\right]}

\newcommand{\argmax}{\arg\max}
\newcommand{\mat}[1]{\mathbf{#1}}
\newtheorem{theorem}{Theorem}
\newtheorem{lemma}[theorem]{Lemma}

\newtheorem{definition}{Definition}

\author{Georgios S. Paschos and Eytan Modiano}
\title{Throughput Optimal Routing in Overlay Networks}

\begin{document}
\maketitle

\begin{abstract}
Maximum throughput requires path diversity enabled by  bifurcating traffic at different network nodes.
In this work, we consider a network where traffic bifurcation is allowed only at  a subset of nodes called \emph{routers}, while the rest nodes (called \emph{forwarders}) cannot bifurcate traffic and hence only forward packets on specified paths.
This implements an overlay network of routers where each overlay link corresponds to a path in the physical network.
We study  dynamic routing implemented at the overlay. 
We develop a queue-based policy, which is  shown to be maximally stable (throughput optimal) for a restricted class of network scenarios
where overlay links do not correspond to overlapping physical paths.
Simulation results show that our policy yields better delay over dynamic policies that allow bifurcation at all nodes, such as the backpressure policy.
Additionally, we provide a heuristic extension of  our proposed overlay routing scheme for the unrestricted class of networks.

%
%
%
%
%

\end{abstract}


\section{Introduction}

A common way to route data in communication networks is shortest path routing. 
Routing schemes using shortest path are \emph{single-path}; they route all packets of a  session through the same dedicated path. 
Although single-path schemes thrive because of their simplicity, they  are in general throughput suboptimal.
Maximizing network throughput requires \emph{multi-path routing}, where the different paths are used to provide diversity \cite{FF62}.

When the network conditions are time-varying or when the session demands fluctuate unpredictably, 
it is required to balance the traffic over  the available paths using a \emph{dynamic routing} scheme which adapts to changes in an online fashion.
In the past,  schemes such as \emph{backpressure} \cite{TE92} have been proposed to discover multiple paths dynamically and mitigate the effects of network variability. 
Although backpressure is desirable in many applications, its practicality is limited by the fact that it requires all nodes in the network to  make online routing decisions.
Often it is the case that some network nodes have limited capabilities and cannot perform such actions.
 \emph{In this paper we study dynamic routing  when decisions can be made only at a subset of nodes, while the rest nodes use fixed single-path routing rules.}

Network overlays are frequently used to deploy new communication architectures in legacy networks~\cite{PetersonDavie}.
To accomplish this, messages from the new technology are encapsulated in the legacy format, allowing the two methods to coexist in the legacy network.
Nodes equipped with  the new technology are then connected in a conceptual network overlay, Fig.~\ref{fig:intro}.
%
Prior  works have considered the use of  this methodology  to introduce new  routing capabilities in the Internet.
For example, content providers use overlays to balance the traffic across different Internet paths and improve resilience and end-to-end performance \cite{Andersen2001,J_Sitaraman_14}.
In our work we use a network overlay  to introduce dynamic routing to a legacy network which operates based on single-path routing.
Nodes that implement the overlay layer are called \emph{routers} and are able to make online routing decisions, bifurcating traffic along different paths. The rest nodes, called \emph{forwarders}, rely on a single-path routing protocol which is available to the physical network, see Fig.~\ref{fig:intro}.

\begin{figure}[t!]
\begin{center}
\includegraphics[width=0.8\columnwidth]{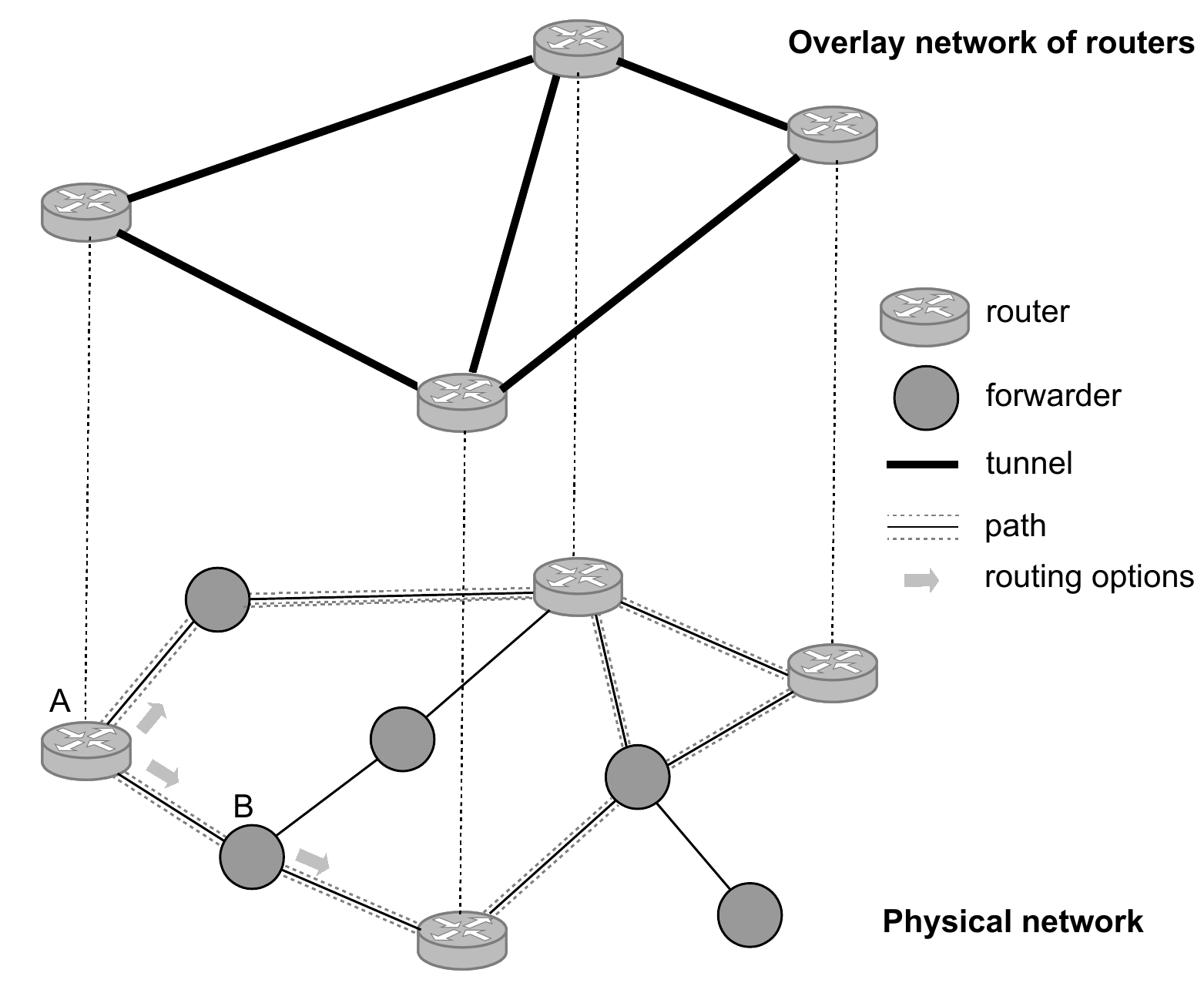}
\end{center}
\vspace{-0.2in}
\caption{ Router A can bifurcate traffic while forwarder B only forwards the packets along a predetermined path.
This paper studies dynamic routing in the overlay.
\vspace{-0.3in}
}
\label{fig:intro}
\end{figure}

There are many applications of our overlay routing model. 
 For  networks with heterogeneous technologies, the overlay routers correspond to devices with extended capabilities, while the  forwarders correspond to less capable devices. 
 For example, to introduce dynamic routing in a network running a legacy routing protocol, it is possible to use Software Defined Networks to install dynamic routing functions on a subset of devices (the routers). 
  In the paradigm of multi-owned networks,  the forwarders are devices where the vendor has no administrative rights.
For example consider a network that uses leased satellite links, where the forwarding rules may be pre-specified by the lease.
In such heterogeneous scenarios,  maximizing throughput by controlling only a fraction of nodes introduces a tremendous  degree of flexibility. 

In the physical network $\mathcal{G}=(\mathcal{N},\mathcal{L})$ denote the set of routers with ${\cal V\subseteq N}$.
Also, denote the throughput region of this network with $\Lambda({\cal V})$ \cite{georgiadis}.\footnote{
The definition of throughput region is given later; here it suffices to think of the set of feasible throughputs.}
Then, $\Lambda({\cal N})$ is the throughput of the network when all nodes are  routers. We call this the full throughput of $\mathcal{G}$, and it can be  achieved if all nodes run the  backpressure  policy \cite{TE92}.
Also, $\Lambda({\cal \emptyset})$ is the throughput of a network consisting only of forwarders, which is equivalent to single-path throughput.
Since increasing the number of routers increases path diversity, we generally have $\Lambda({\cal \emptyset})\subseteq \Lambda({\cal V})\subseteq \Lambda({\cal N})$.
Prior work studies the necessary and sufficient conditions for router set 
$\cal V^*$ to  guarantee full throughput, i.e., $\Lambda({\cal V^*})=\Lambda({\cal N})$ \cite{C_jones_14}.
 The results of the study show that using a small percentage of routers ($8\%$) is sufficient  for full throughput  in power-law random graphs--an accurate model of the Internet \cite{NewmanBook}. 
 Although  \cite{C_jones_14}  characterizes the throughput region  $\Lambda({\cal V})$,  a    dynamic routing to achieve this performance is still unknown.
 For example, in the same work it is showcased that backpressure operating in the overlay  is suboptimal.
 \emph{In this work we fill this gap under a specific topological  assumption explained in detail later. We study dynamic routing in the overlay network of  routers  and propose a control policy that achieves $\Lambda({\cal V})$. Our work is the first to analytically study such a heterogeneous dynamic routing policy and prove its optimality.}

\section{System Model}

We consider a physical network $\mathcal{G}=(\mathcal{N},\mathcal{L})$ where the nodes are partitioned to  routers  ${\cal V}$ and forwarders ${\cal N-V}$.
The physical network has installed single-path routing rules,
which we capture as follows.
Every router $i\in{\cal V}$ is assigned an acyclic path $p_{ij}$ to every other router $j\in{\cal V}$.\footnote{The legacy routing protocol may provide paths between physical nodes as well, but we do not study them in this work.}
 Fig.~\ref{fig:model} (left) shows with bold arrows both paths assigned to router $\textsf{a}$
 , i.e., $(\textsf{a,d,e})$,  and $(\textsf{a,b,c})$. Let $P$ be the set of all such  paths in the network.

\begin{figure}[t!]
\begin{center}
\includegraphics[width=0.46\columnwidth]{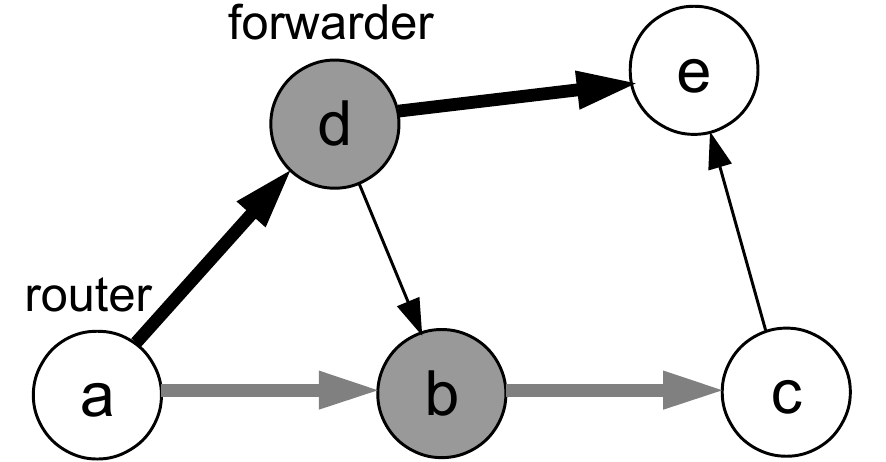}
\includegraphics[width=0.46\columnwidth]{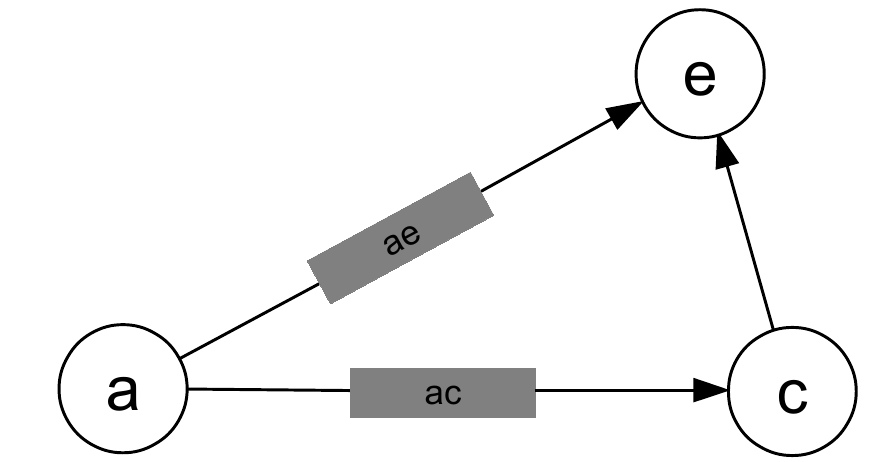}
\end{center}
\caption{(left) An example network of routers and forwarders, where routers are ${\cal V}=\{\textsf{a,c,e}\}$. We indicate with bold arrows the shortest paths available to \textsf{a} by the  single-path routing scheme of the physical network.  (right) The equivalent overlay network of routers and tunnels. 
}
\vspace{-0.35in}
\label{fig:model}
\end{figure}

\subsection{The Overlay Network of Tunnels}

We introduce the concept of \emph{tunnels}.
The  tunnel $(i,j)\in{\cal E}$ corresponds to a path $p_{ij}\in P$ with end-points routers $i,j$ and intermediate nodes  forwarders.
We then  define the overlay network  $\mathcal{G}_R=(\mathcal{V},\mathcal{E})$ consisting of routers $\mathcal{V}$  and tunnels $\mathcal{E}$. 
Figure~\ref{fig:model} (right) depicts  the overlay network  for the physical network in the left, assuming shortest path routing is used.

\subsubsection{Topological Assumption}
\label{sec:overlaping}

In this work we study the case of \emph{non-overlapping tunnels}. 
Let ${\cal T}_{ij}$ be the set of all physical links of tunnel $(i,j)$ with the exception of the first input link.
\begin{definition}[Non-Overlapping Tunnels]
An overlay network  satisfies the non-overlapping tunnels condition if for any two tunnels $e_1\neq e_2$ we have ${\cal T}_{e_1}\cap {\cal T}_{e_2}=\emptyset$. 
\end{definition}
Whether the condition is satisfied or not, depends on the network topology ${\cal G}$, the set of routers ${\cal V}$, and the set of paths $P$ which altogether determine ${\cal T}_{ij}$, for all $i,j\in {\cal V}$.
The network of Figure \ref{fig:model} satisfies the non-overlapping tunnels condition since each of the links $\textsf{(d,e)}, \textsf{(b,c)}$  belongs to exactly one tunnel. On the other hand, in the network of Figure \ref{fig:model2}  link  $\textsf{(c,d)}$    belongs to two  tunnels, hence the condition is not satisfied.

When tunnels overlap,
packets belonging to different tunnels compete for service at
the forwarders, which further complicates the analysis.
Our analytical results  focus exclusively on the non-overlapping tunnels case which still constitutes an interesting and difficult problem.
However, in the simulation section we heuristically extend our proposed policy to apply to general networks with overlapping tunnels  and showcase that the extended policy has near-optimal performance.

\begin{figure}[t!]
\begin{center}
\includegraphics[width=0.46\columnwidth]{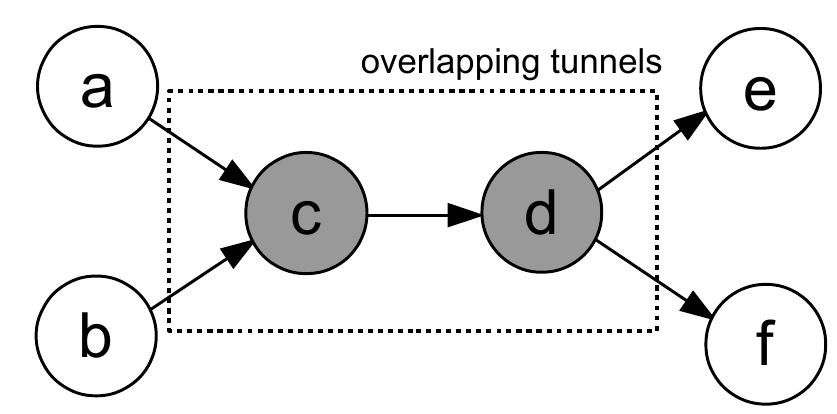}
\includegraphics[width=0.46\columnwidth]{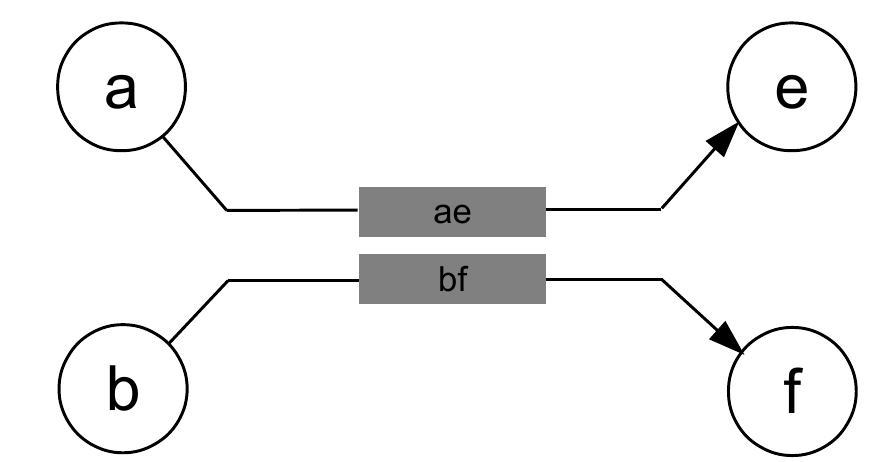}
\end{center}
\vspace{-0.2in}
\caption{An example with overlapping tunnels.
}
\label{fig:model2}
\end{figure}

\subsection{Overlay Queueing Model}

The overlay network admits a  set of sessions $\mathcal C$, where
 each session has a unique router destination, but possibly multiple router sources.
Time is slotted; at the end of time slot $t$,  $A_i^c(t)\leq A_{\max}$ packets of session $c\in{\cal C}$ arrive exogenously at router $i$, where $A_{\max}$ is a positive constant. 
\footnote{Note that we focus exclusively on routing at the overlay layer. Thus $A_i^c(t)$ are defined at overlay router nodes. }
$A_i^c(t)$ are i.i.d. over slots, independent across sessions and sources, with mean 
$\lambda_i^c$. 

For every tunnel $(i,j)$, a routing policy $\pi$ chooses  the  routing function $\mu_{ij}^c(t,\pi)$ in slot $t$ which determines \emph{the number of session $c$ packets} to be routed  from router $i$ into the tunnel.
Additionally, we denote with $\phi_{ij}^c(t)$ the actual number of session $c$ packets that exit the tunnel in slot $t$.
For a visual association of $\mu_{ij}^c(t,\pi)$ and $\phi_{ij}^c(t)$ to the tunnel links see Figure \ref{fig:functions}.
Note that  $\mu_{ij}^c(t,\pi)$ is decided by router $i$ while  $\phi_{ij}^c(t)$ is uncontrollable.
\begin{figure}[t!]
\begin{center}
\begin{overpic}[scale=.6]{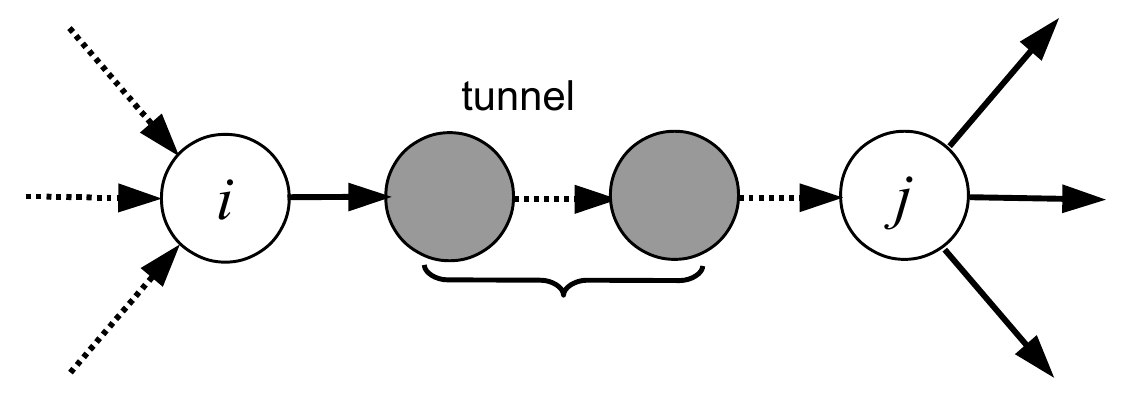}
\put(52,26){\footnotesize $(i,j)$}
\put(26,22){\footnotesize $\mu_{ij}$}
\put(25,13){\footnotesize $R_{ij}^{\text{in}}$}
\put(67,22){\footnotesize$\phi_{ij}$}
\put(17.5,7){\footnotesize$Q_i$}
\put(78,7){\footnotesize$Q_j$}
\put(47,5){\footnotesize$F_{ij}$}
\end{overpic}
\caption{The input of a tunnel is controllable (solid line) but the output is uncontrollable (dotted line).}
\vspace{-0.3in}
\label{fig:functions}
\end{center}
\end{figure}
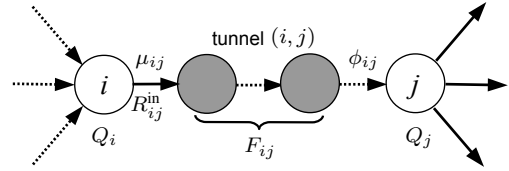

Let the sets $\text{In}(i), \text{Out}(i)$ represent the  incoming and outgoing neighbors of router $i$ on ${\cal G}_R$.
Packets of session $c$ are stored at router $i$ in a \emph{router queue}. 
Its backlog $Q_i^c(t)$  evolves according to the following equation
\begin{align}\label{eq:Qevol}
Q_i^c(t+1)=\Big(Q_i^c(t)-\hspace{-0.12in}\underbrace{\sum_{b\in \text{Out}(i)}\!\!\!\! \mu_{ib}^{c}(t,\pi)}_{\text{departures}}\Big)^+\!\!\!+\hspace{-0.12in}\underbrace{\sum_{a\in \text{In}(i)}\!\! \phi_{ai}^{c}(t)+A_i^c(t)}_{\text{arrivals}},\raisetag{0.5\baselineskip}
\end{align}
where 
 we use $(.)^+\triangleq \max\{.,0\}$ since there might not be enough packets to transmit.

On tunnel $(i,j)$ we collect all packets into one \emph{tunnel queue} $F_{ij}(t)$ whose evolution satisfies 
%
\begin{align}\label{eq:queueFij}
F_{ij}(t+1)\leq F_{ij}(t)-\underbrace{\sum_c\phi_{ij}^{c}(t)}_{\text{departures}}+\underbrace{\sum_c\mu_{ij}^{c}(t,\pi)}_{\text{arrivals}}, ~\forall (i,j)\in\mathcal{E}.\raisetag{1\baselineskip}
\end{align}
The packets that actually arrive at $F_{ij}(t)$ might be less than $\sum_c\mu_{ij}^{c}(t,\pi)$, hence the inequality  (\ref{eq:queueFij}). 
%
 We remark that $F_{ij}(t)$ is the total number of packets in flight on the tunnel $(i,j)$. 
 Physically these packets are stored at different forwarders along the tunnel. 
 We only keep track of the sum of these physical backlogs since, as we will show shortly, this is sufficient to achieve maximum throughput.

Above  (\ref{eq:Qevol})  assumes that all incoming traffic at router $i$ arrives either from tunnels, or exogenously. It is possible, however, to have an incoming neighbor router $k$ such that $(k,i)$ is a physical link,
a case we purposely omitted in order  to avoid further complexity in the exposition. 
The optimal policy  for this case  can be obtained from our proposed policy by setting the corresponding tunnel queue backlog to zero,  $F_{ki}(t)=0$.

\subsection{Forwarder Scheduling Inside Tunnels}
We assume that inside tunnels  packets are forwarded in a \emph{work-conserving} fashion, i.e., a forwarder does not idle unless there is nothing to send.
Due to work-conservation and the assumption of non-overlapping tunnels,  a tunnel with ``sufficiently many'' packets  has instantaneous  output  equal to its bottleneck capacity. 
Denote by $M_{ij}$  the number of forwarders associated with tunnel $(i,j)$.
Let $R_{ij}^{\max}$ be the greatest capacity among all  physical links associated with tunnel  $(i,j)$
and $\pacap$  the smallest, also let
\begin{equation}\label{eq:T0}T_0\triangleq \max_{(i,j)\in{\cal E}} \left[M_{ij}\pacap+\frac{M_{ij}(M_{ij}-1)}2R_{ij}^{\max}\right].\end{equation}
\begin{lemma}[Output of a Loaded Tunnel]\label{lem:leaky}
Under any control policy $\pi\in\Pi$,  suppose that in time slot $t$ the total tunnel backlog satisfies $ F_{ij}(t)>T_0$, for some $(i,j)\in{\cal E}$, where $T_0$ is defined in \eqref{eq:T0}. 
The instantaneous output of the tunnel satisfies 
\begin{equation}\label{eq:leaky}
\sum_c\phi_{ij}^c(t)=\pacap.
\end{equation}
\end{lemma}
\begin{proof} The proof  is provided in the Appendix~\ref{app:lem}.\end{proof}
Lemma \ref{lem:leaky}  is a path-wise statement saying that  the tunnel output is equal to the tunnel bottleneck capacity in every time slot that the tunnel backlog exceeds $T_0$.

Notably we haven't discussed yet how the forwarders choose to prioritize packets from different sessions.
Based on Lemma  \ref{lem:leaky} and the results that follow,
 we will establish that independent of the choice of session scheduling policy, there exists a routing policy that maximizes throughput. 
 Furthermore, we  demonstrate by simulations that different forwarding scheduling policies result in the same average delay performance under our proposed routing.
%
Hence, in this paper forwarders are  allowed to use any work-conserving   session scheduling, such as
  FIFO, Round Robin or even strict priorities among sessions.

\section{Dynamic Routing Problem Formulation}

A choice for the routing function $\mu_{ij}^c(t,\pi)$ is considered permissible if it satisfies in every slot the corresponding capacity constraint $\sum_c\mu_{ij}^c(t,\pi)\leq \Cin$, where $\Cin$ denotes the capacity of the input physical link  of tunnel  $(i,j)$, see Fig.~\ref{fig:functions}.
In every time slot, a control policy $\pi$ determines the routing functions $\left(\mu_{ij}^c(t,\pi)\right)$ at every router. Let $\Pi$ be the class of all permissible control policies, i.e., the policies whose sequence of decisions consists of permissible routing functions.

We want to  keep the backlogs small in order to guarantee that the throughput is equal to the arrivals. To keep track of this we define the stability criterion adopted from \cite{georgiadis}.

\begin{definition}[System Stability]
A queue  with backlog $X(t)$  is stable under policy $\pi$ if 
\[
\limsup_{T\to\infty}\frac{1}T\sum_{t=0}^{T-1}\mean{X(t)}<\infty.
\]
The overlay network is stable if all router $(Q_i^c(t))$ and tunnel queues $(F_{ij}(t))$ are stable.
\end{definition}

The \emph{throughput region} $\Lambda({\cal V})$ of class $\Pi$ is defined to be (the closure of) the set of  $\boldsymbol\lambda= (\lambda_i^c)$ for which there exists a policy $\pi\in\Pi$ such that the system is stable. 
Avoiding technical jargon, the throughput region includes all achievable throughputs when implementing dynamic routing in the overlay.
Recall that  throughput depends on the actual selection of routers ${\cal V}$, and that for ${\cal V\subset N}$ it may be the case that
the achievable throughput may be less than the full throughput of ${\cal G}$, i.e., 
 $\Lambda({\cal V})\subset \Lambda({\cal N})$. 
Therefore it is important to clarify that in this work we assume that ${\cal V}$ is fixed and we seek to find a policy that is stable for any  $\boldsymbol\lambda\in\Lambda({\cal V})$, i.e., a policy that is \emph{maximally stable}. 
Such a policy is also called in the literature ``throughput optimal''.

\subsection{Characterization of Throughput Region of Class $\Pi$}\label{sec:region}

The throughput region $\Lambda({\cal V})$ can be characterized as  the  closure of the  set of  matrices $\boldsymbol\lambda = (\lambda_i^c)$ for which there exist nonnegative flow variables  $(f_{ij}^c)$ such that
\begin{eqnarray}
&& \lambda_i^c + \sum_{a\in\mathcal{V}} f_{ai}^c < \sum_{b\in\mathcal{V}} f_{ib}^c, ~~~~ \text{ for all } i\in \mathcal{V}, c\in \mathcal C \label{eq:region1}\\
&& \sum_c f_{ij}^c<\pacap, \hspace{0.615in}  \text{ for all } (i,j), \in \mathcal E,\label{eq:region2}
\end{eqnarray}
where (\ref{eq:region1}) are flow conservation inequalities at routers, (\ref{eq:region2}) are capacity constraints on tunnels, and recall that $\pacap$ is the bottleneck capacity in the tunnel $(i,j)$.
We write
\[
\Lambda({\cal V})=\text{Cl}\{\boldsymbol\lambda~|~\boldsymbol f\geq \boldsymbol 0, \text{ and \eqref{eq:region1}-\eqref{eq:region2} hold} \}.
\]
Note, that the conditions for the stability region $\Lambda(\mathcal{V})$ are the same with the conditions for full throughput $\Lambda(\mathcal{N})$~\cite{georgiadis}, with the difference that the flow variables are defined on the network of routers ${\cal G}_R$ instead of $\cal{G}$. 
Indeed the proof that (\ref{eq:region1})-(\ref{eq:region2}) are necessary and sufficient for stability may be obtained by considering a virtual network where every tunnel is replaced by a virtual link.

Controlling this system in a dynamic fashion amounts to finding a routing policy $\pi^*\in\Pi$ which stabilizes the system for any $\boldsymbol \lambda\in\Lambda({\cal V})$.
\emph{Finding such a policy in the overlay differs significantly from the case of a physical network}, since physical links support immediate transmissions while overlay links are   work-conserving tandem queues which  induce queueing delays.

\section{The Proposed Routing Policy}
\label{sec:main}

As discussed in \cite{C_jones_14}, using backpressure in the overlay may result in poor throughput performance.
In this section we propose the Threshold-based Backpressure ($\TB$) Policy, a distributed policy which performs  online decisions in the overlay.  
$\TB$ is designed to operate the tunnel backlogs close to a threshold. 
This is a delicate balance whereby the tunnel output works efficiently (by Lemma \ref{lem:leaky}) while at the same time the number of packets in the tunnel are upper bounded. 

Consider the threshold
\begin{equation}\label{eq:thres}
T=T_0+\max_{(i,j)}\Cin,
\end{equation}
where $T_0$ is defined in (\ref{eq:T0}) and $\Cin$ 
is the  capacity of input physical link of tunnel $(i,j)$ and thus also the maximum increase of the tunnel backlog in one slot. 
Define the condition:
\begin{equation}\label{eq:C1}
F_{ij}(t)\leq T.
\end{equation}
The reason we use this threshold is that if (\ref{eq:C1}) is false, it follows that both $F_{ij}(t)>T_0$ and $F_{ij}(t-1)>T_0$, and hence we can apply Lemma 1 to both slots $t$ and $t-1$. This is used in the  proof of the main result.

\noindent \rule[0.05in]{3.5in}{0.01in} 

\vspace{-0.03in}
\begin{centering} \textbf{Threshold-based Backpressure ($\TB$) Policy }

\vspace{-0.03in}
\end{centering}
\noindent \rule[0.05in]{3.5in}{0.01in}

At each time slot $t$ and tunnel $(i,j)$, 
 let 
\[
c_{ij}^*\in \argmax_{c\in{\cal C}} Q_i^c(t)-Q_j^c(t),
\]
be a session that maximizes the differential backlog between  routers  $i,j$, ties resolved arbitrarily. 
 Then route into that tunnel
\begin{equation}\label{eq:servf}
\mu_{ij}^{c_{ij}^*}(t,\text{TB})=\left\{\begin{array}{ll}
\Cin & \text{ if  } Q_i^{c_{ij}^*}(t)>Q_j^{c_{ij}^*}(t) \\
& \text{ AND } (\ref{eq:C1}) \text{ is true}\\
& \\
0 & \text{ otherwise} 
\end{array}\right.
\end{equation}
and $\mu_{ij}^{c}(t,\text{\TB})=0,~\forall c\neq c_{ij}^*$. Recall, that $\Cin$ denotes the  capacity of input physical link of tunnel $(i,j)$.
\footnote{If the there are not enough packets to transmit, i.e., $\mu_{ij}^{c_{ij}^*}(t)>Q_i^{c_{ij}^*}(t)$, then we fill the transmissions with dummy non-informative packets.} 

\noindent \rule[0.05in]{3.5in}{0.01in}

$\TB$ is similar to applying backpressure in the overlay, with the striking difference that \emph{no packet is transmitted to a tunnel} if  condition (\ref{eq:C1}) is not satisfied. Therefore the total  tunnel backlog is limited to at most $T$ plus the maximum number of packets that may enter the tunnel in one slot. Formally we have 
\begin{lemma}[Deterministic bounds of $F_{ij}(t)$ under $\TB$]\label{lem:detFb}
Assume that the system starts empty and is operated under $\TB$. Then the tunnel backlogs $\left(F_{ij}(t)\right)$ are uniformly bounded above by 
\begin{equation}\label{eq:Fmax}
F^{\max}\triangleq T+R_{\max}.
\end{equation}
\end{lemma}
\begin{proof} Follows from \eqref{eq:C1} and \eqref{eq:servf}.
\end{proof}
This shows that our policy does not allow the tunnel backlogs to grow beyond $F^{\max}$. To show that our policy efficiently routes the packets is much more involved. It is included  in the proof of the following main result.



\begin{theorem}\label{th:optimality}[Maximal Stability of $\TB$]
\label{th:opti}
Consider an overlay network where 
 underlay forwarding nodes use any work-conserving policy to schedule packets over predetermined paths,  
and  the tunnels are non-overlapping.

The $\TB$ policy is maximally stable:
\[
\Lambda^{\TB}({\cal V})\supseteq \Lambda^{\pi}({\cal V}),~\text{for all}~\pi\in\Pi.
\]
\end{theorem}
\begin{proof}{The proof is 
is based on a novel $K$-slot Lyapunov drift analysis and 
it  is given in the Appendix~\ref{app:th}.}\end{proof}
$\TB$ is a distributed  policy since it utilizes only local queue information and the capacity of the incident links, while it is agnostic to arrivals, or  capacities of remote links, e.g. note that the decision  does not depend on the capacity of the bottleneck link  $\pacap$. 

{ A very simple distributed protocol can be used to allow overlay nodes to learn the tunnel backlogs. 
Specifically  $F_{ij}(t)$  can be estimated at node $i$ using an acknowledgement scheme, whereby $j$ periodically informs $i$ of how many packets have been received so far. 
 In practice, the router nodes obtain a delayed estimate  $\tilde{F}_{ij}(t)$. However, using the concepts in \cite{B_neely_10}-p.85, it is possible to  show that  such estimates do not hurt the efficiency of the scheme.}

\section{Simulation Study}

In this section we perform extensive simulations  to:
\begin{itemize}
\item[(i)] showcase the maximal stability of $\TB$ and compare its throughput performance to other routing policies,
\item[(ii)] examine the impact of  different  forwarding scheduling policies (FIFO, HLPSS, Strict Priority, LQF) on throughput and delay of $\TB$,
\item[(iii)] demonstrate that $\TB$ has good delay performance, and
\item[(iv)] study the extension of $\TB$ to the case of overlapping tunnels.
\end{itemize}

First we present  dynamic routing  policies from the literature against which we will compare $\TB$.

\textbf{Backpressure in the overlay ($\BPr$):} For every tunnel  $(i,j)\in{\cal E}$ define 
\[
c_{ij}^*\in \argmax_{c\in{\cal C}} Q_i^c(t)-Q_j^c(t),
\]
ties solved arbitrarily. Then  choose   $\mu_{ij}^{c}(t,\BPr)=0, c\neq c_{ij}^*$ and
\begin{equation*}\label{eq:servBP}
\mu_{ij}^{c_{ij}^*}(t,\BPr)=\left\{\begin{array}{ll}
\Cin & \text{ if  } Q_i^{c_{ij}^*}(t)>Q_j^{c_{ij}^*}(t) \\
0 & \text{ otherwise.} 
\end{array}\right.
\end{equation*}
 This corresponds to backpressure applied only to routers  ${\cal V}$, which is  admissible in our system, $\BPr\in\Pi$.

\textbf{Backpressure in the physical network ($\BPa$):} For every physical link  $(m,n)\in{\cal L}$ define  
\[
c_{mn}^*\in \argmax_{c\in{\cal C}} Q_m^c(t)-Q_n^c(t)
\]
ties solved arbitrarily. Then  choose   $\mu_{mn}^{c}(t,\BPa)=0, c\neq c_{mn}^*$ and
\begin{equation}\label{eq:servBP}
\mu_{mn}^{c_{mn}^*}(t,\BPa)=\left\{\begin{array}{ll}
R_{mn} & \text{ if  } Q_m^{c_{mn}^*}(t)>Q_n^{c_{mn}^*}(t) \\
0 & \text{ otherwise} 
\end{array}\right.
\end{equation}
This is the classical backpressure from \cite{TE92}, applied to all nodes ${\cal N}$ in the network, and thus it is not admissible in the overlay, $\BPa\notin\Pi$, whenever ${\cal V\subset N}$. 
Since this policy achieves the full throughput $\Lambda({\cal N})$, we use it as a throughput benchmark.

\textbf{{Backpressure Enhanced with Shortest Paths Bias ($\BPs$):}} 
For every node-session pair $(m,c)$ define the hop count from $m$ to the destination of $c$ as $h_n^c$. For every physical link  $(m,n)\in{\cal L}$ define
\[
c_{mn}^*\in \argmax_{c\in{\cal C}} Q_m^c(t)-Q_n^c(t)+h_m^c-h_n^c.
\]
ties solved arbitrarily. Then choose $\mu_{mn}^c(t,\BPs)$ according to (\ref{eq:servBP}). This policy was proposed by \cite{NMR05} to reduce delays.
When the congestion is small, the shortest path bias introduced by the hop count difference leads the packets directly to the destination without going through  cycles or longer paths.
Such a policy requires control at every node, and thus it is not admissible in the overlay, $\BPs\notin\Pi$, whenever ${\cal V\subset N}$. Since, however, it is known to achieve $\Lambda({\cal N})$ and to outperform $\BPa$ in terms of delay, it is useful for throughput and  delay comparisons.


%
%

\subsection{Showcasing Maximal Stability}

Consider the network of  Figure \ref{fig:maxsta} (left), and define two sessions sourced at $\textsf{a}$; session 1  destined to $\textsf{e}$ and session 2 to $\textsf{c}$. 
We assume that $R_{\textsf{ab}}=2$ and all the other link capacities are unit as shown in the Figure.
 We choose $R_{\textsf{ab}}$ in this way to make the routing decisions of session 1 more difficult.
  We show the full throughput region $\Lambda({\cal N})$ achieved by $\BPa,\BPs$ which however are not admissible in the overlay. 
Then we experiment with $\TB,\BPr$ and we also show the throughput of plain Shortest Path routing.
For $\TB$, according to example settings and (\ref{eq:thres}) it is $T_0=2$; we choose $T=6$. 

Since the example satisfies the non-overlapping tunnel condition, by Theorem \ref{th:opti} our policy achieves $\Lambda({\cal V})$.
This is verified in the simulations, see Figure~\ref{fig:maxsta} (right). 
From the figure we can conclude that for this example we have $\Lambda({\cal V})=\Lambda({\cal N})$, although ${\cal V\subset N}$.
This is consistent to the findings of \cite{C_jones_14}.
From the same Figure we see that both backpressure in the overlay $\BPr$ and Shortest Path achieve only a fraction of $\Lambda({\cal V})$, and hence they are not maximally stable. 
For $\BPr$, we have loss of throughput when both sessions compete for traffic, in which case $\BPr$ fails to consider congestion information from the tunnel $\textsf{ac}$ and therefore  allocates this tunnel's resources wrongly to the two sessions.
For Shortest Path, it is clear that each session uses only its own dedicated shortest path and hence the loss of throughput is due to no path diversity.

\begin{figure}[t!]
\begin{center}
\begin{overpic}[scale=.4]{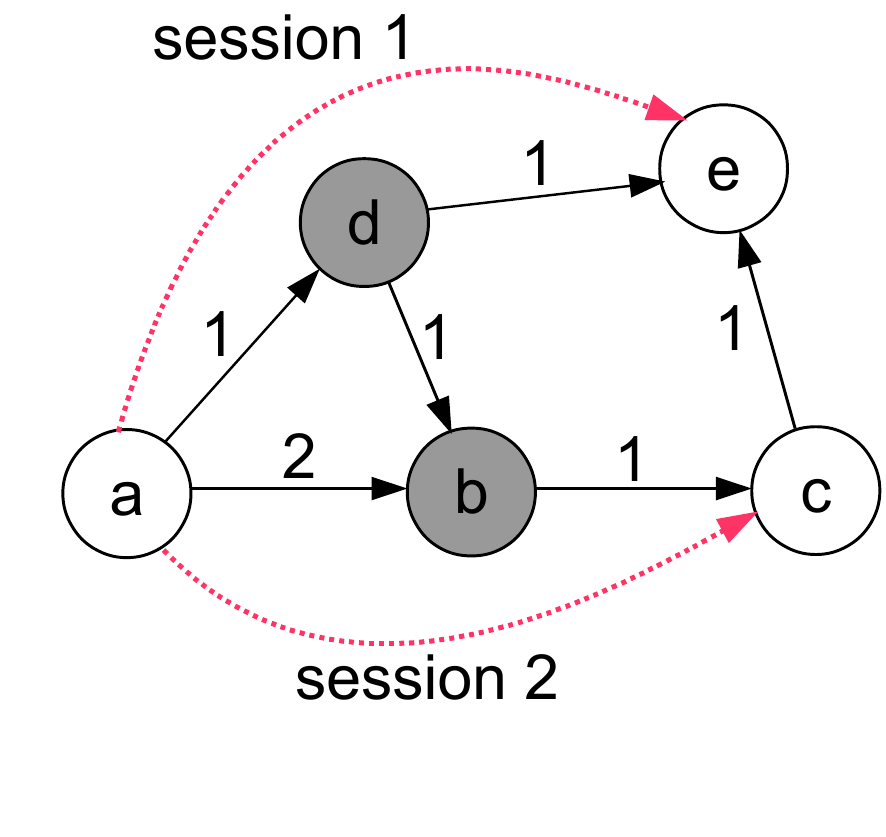}
\end{overpic}~~~
\begin{overpic}[scale=.47]{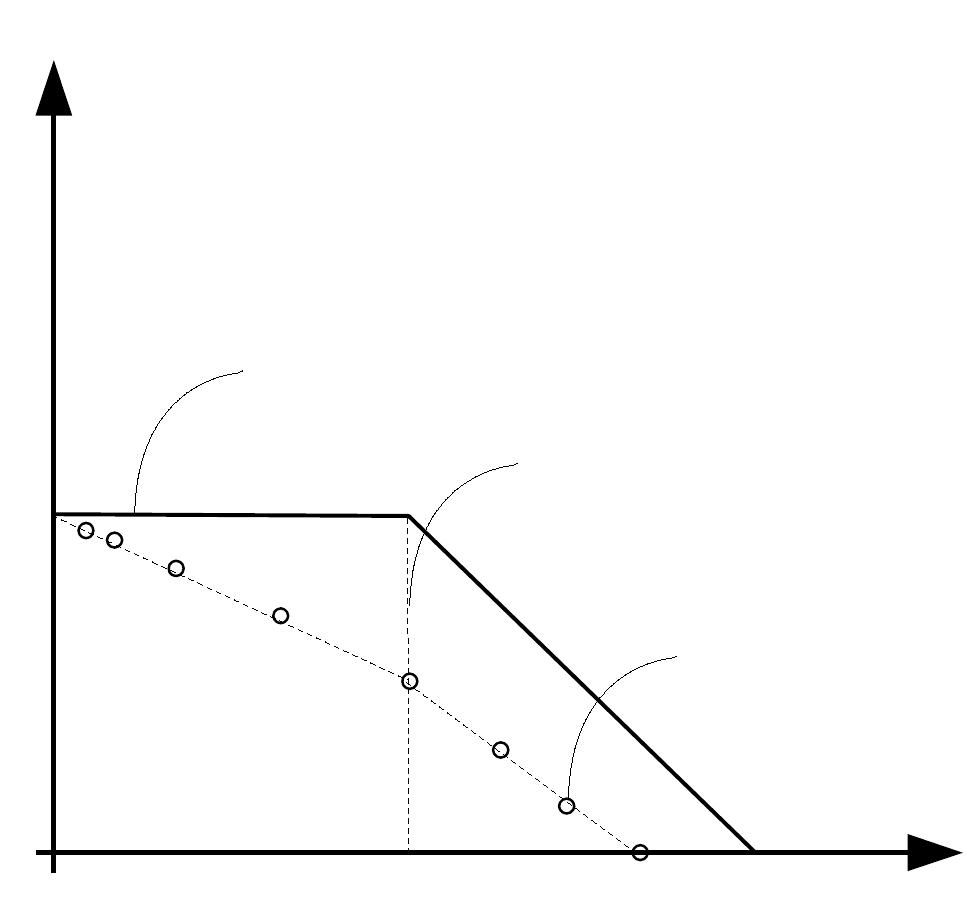}
\put(87,0){\small$\lambda_1$}
\put(-4,80){\small$\lambda_2$}
\put(27,55){\footnotesize$\TB$,$\overbrace{\BPa,\BPs}^{\text{not admissible in $\Pi$}}$}
\put(71,26){\footnotesize $\BPr$}
\put(54,45.3){\footnotesize Shortest Path}
\put(70,0){\footnotesize $(2,0)$}
\put(37,0){\footnotesize $(1,0)$}
\put(0,0){\footnotesize $\mathbf{0}$}
\put(-10,40){\footnotesize $(0,1)$}
\end{overpic}
   \caption{Throughput comparison: (left) Example under study. (right) Throughput achieved by $\{\TB, \BPr, \text{Shortest Path}\}\subset \Pi$ and $ \BPa, \BPs\notin \Pi$.}
     \label{fig:maxsta}
\end{center}
\end{figure}
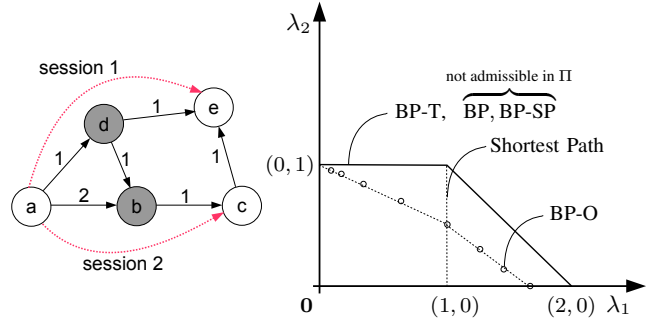
 

To understand why $\TB$ works,
we  examine a sample path evolution of this system under  $\TB$ for the case where $\lambda_1=\lambda_2=0.97$, which is one of the most challenging scenarios.
For stability, session 1  must use its dedicated path $(\textsf{a,d,e})$, and send almost no traffic through tunnel $\textsf{ac}$.
Focusing on the tunnel $\textsf{ac}$, Figure \ref{fig:samplepath1} shows the differential backlogs per session $Q_{\textsf{a}}^c(t)-Q_{\textsf{c}}^c(t)$ and the corresponding tunnel backlog $F_{\textsf{a}\textsf{c}}(t)$ for a sample path of the system evolution.  
In most time slots $\textsf{a}$  is congested, which is indicated by high differential backlogs. In such slots, the tunnel has more than 1 packet, which guarantees by Lemma \ref{lem:leaky} that it outputs packets at highest possible rate, hence the tunnel is correctly utilized.
Recall that when the tunnel is  full ($F_{\textsf{a}\textsf{c}}(t)>T$=6) no new packets are inserted to the tunnel preventing it from exceeding $F_{\max}$.
Observe that the differential backlog of session 2 always dominates the session 1 counterpart, and hence whenever a tunnel is again ready for a new packet insertion, session 2 will be prioritized for transmission according to (\ref{eq:servf}). Therefore, the proportion of session 2 packets in this tunnel is close to 100\%, which is the correct allocation of the tunnel resources to sessions for this case. 

\begin{figure}[t!]
\begin{center}
\begin{overpic}[scale=.2]{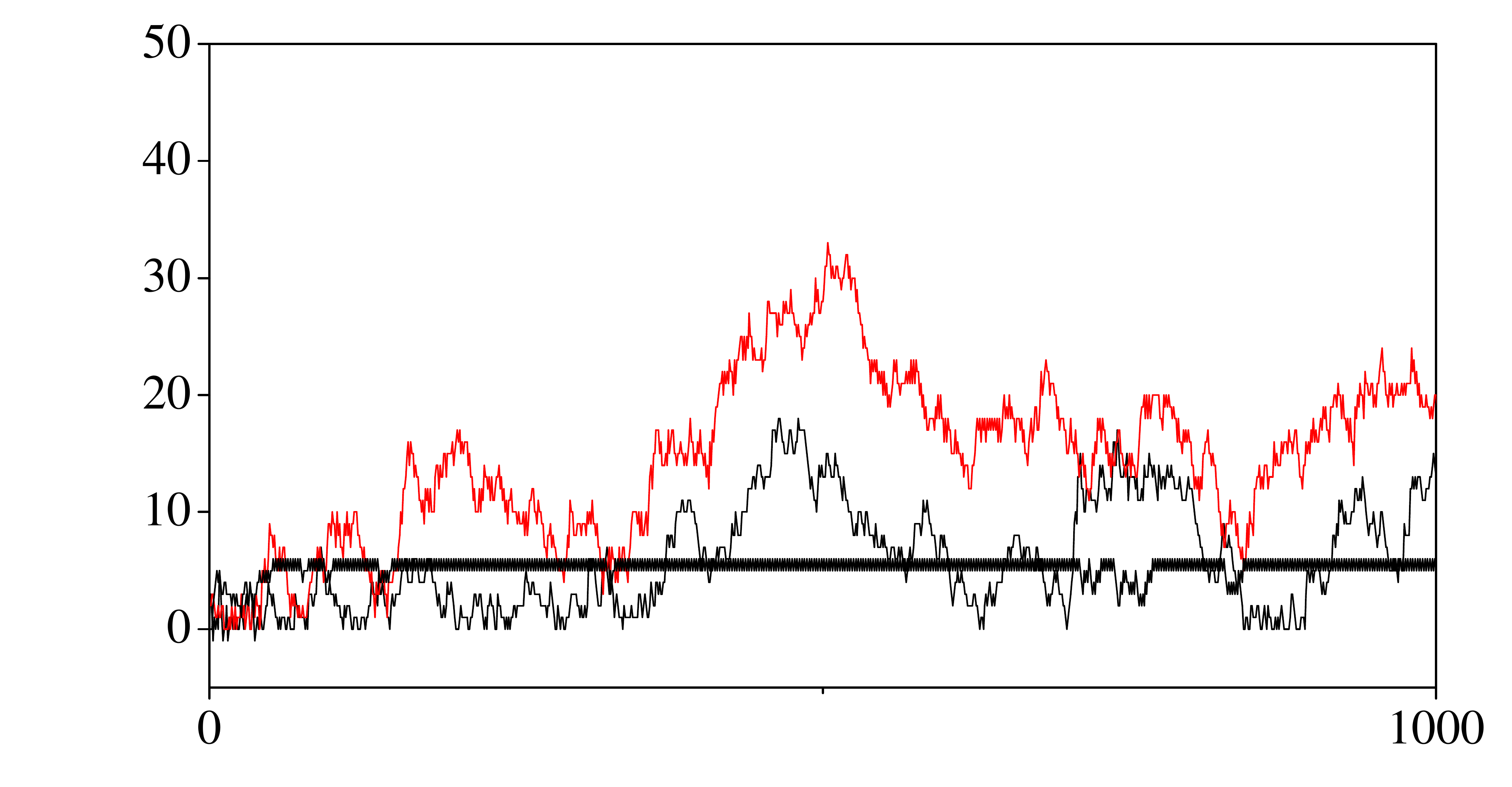}
\put(65,30){\footnotesize $Q_{\textsf{a}}^1(t)-Q_{\textsf{c}}^1(t)$}
\put(57,44){\footnotesize $Q_{\textsf{a}}^2(t)-Q_{\textsf{c}}^2(t)$}
\put(29,42){\footnotesize $F_{\textsf{a}\textsf{c}}(t)$}
\put(40,2){\small time (slots)}
\put(36,50){\small individual backlogs}
\put(31,40){\color{black}\vector(0,-1){24.5}}
\put(59,42){\color{black}\vector(-1,-3){2.5}}
\put(64,30){\color{black}\vector(-1,-2){6}}
\end{overpic}
   \caption{Sample path evolution of the system under $\TB$, $\lambda_1=\lambda_2=.97$.}
\vspace{-0.4in}
     \label{fig:samplepath1}
\end{center}
\end{figure}
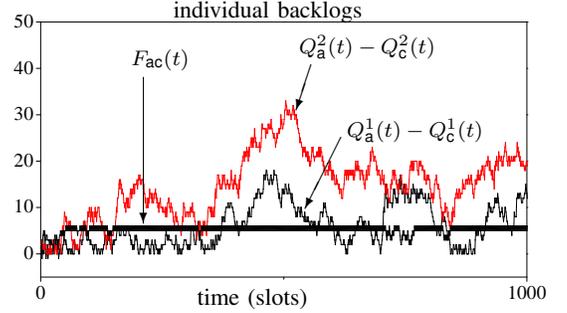

\subsection{Insensitivity to Forwarding Scheduling}

At every forwarder node there is a packet scheduling decision to be made,
 to choose how many packets per session should be forwarded in the next slot. 
 Although by assumption we require the forwarding policy to be work-conserving, 
 our results do  not restrict the scheduling  policy any further. 
In particular, our analysis only depends on $\sum_c\phi_{ij}^{c}(t)$ and hence it is insensitive to  the chosen discipline.

Here we simulate the operation of $\TB$ with different forwarding policies, in particular with First-In First-Out (FIFO), Head of Line Proportional Processor Sharing (HLPPS), Strict Priority and Longest Queue First (LQF), where HLPPS refers to serving sessions proportionally to their queue backlogs \cite{hlpps}, and LQF refers to giving priority to the session with the longest queue.
Figure \ref{fig:fwddiscipline} shows sample path differences for several  forwarding disciplines on the example of the previous section, while Table~\ref{delay_table} compares the average delay performance for different arrival rates.
Independent of the discipline used, the average total number of packets in the system is approximately the same. 
Therefore, while our theorem states that the forwarding policy does not affect  $\TB$ throughput, simulations additionally show that the delay  is also the same.

\begin{figure}[t!]
\begin{center}
\begin{overpic}[scale=.225]{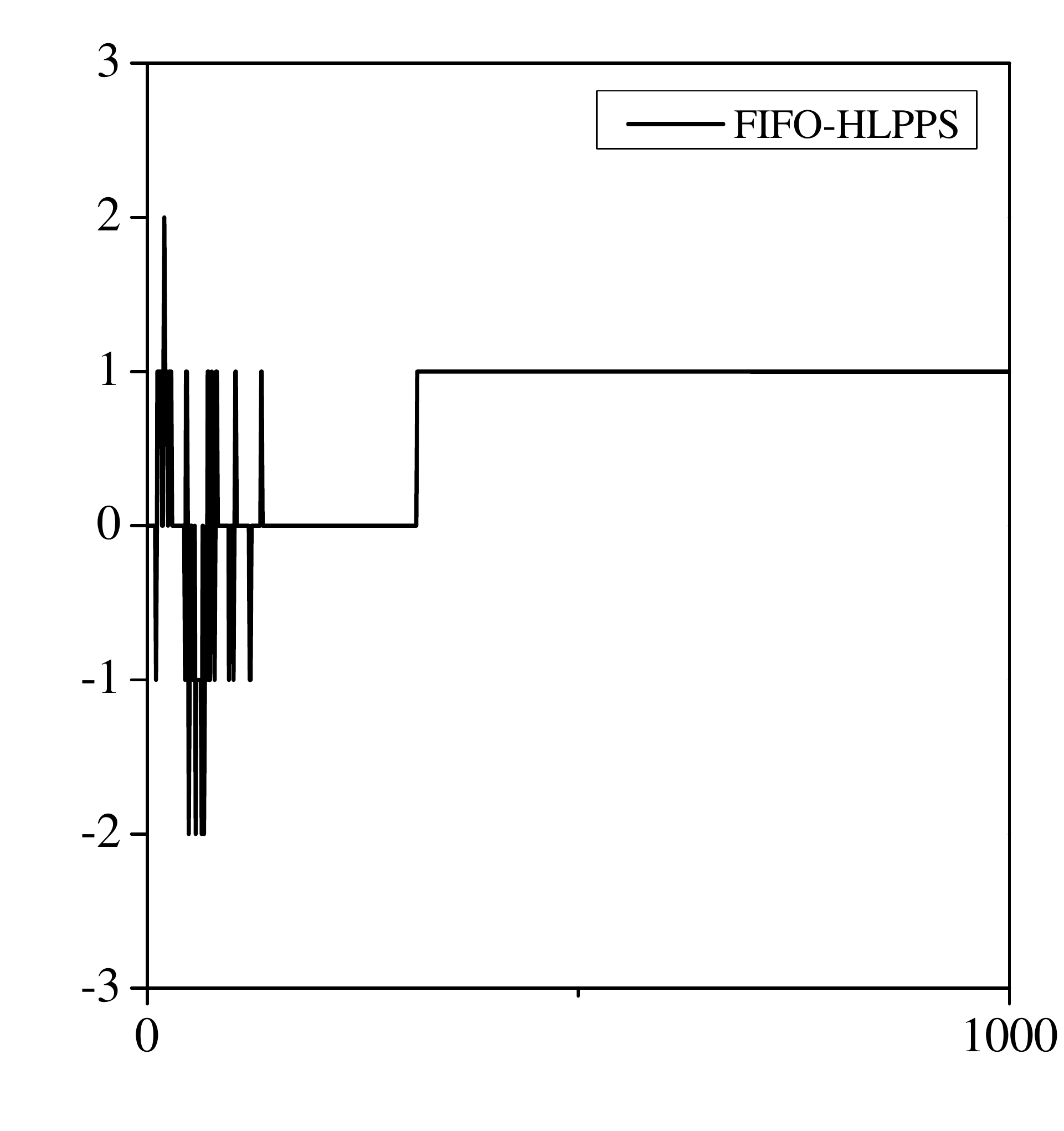}
\put(40,5){\small time (slots)}
\put(19,96){\small total backlog difference}
\end{overpic}
\begin{overpic}[scale=.225]{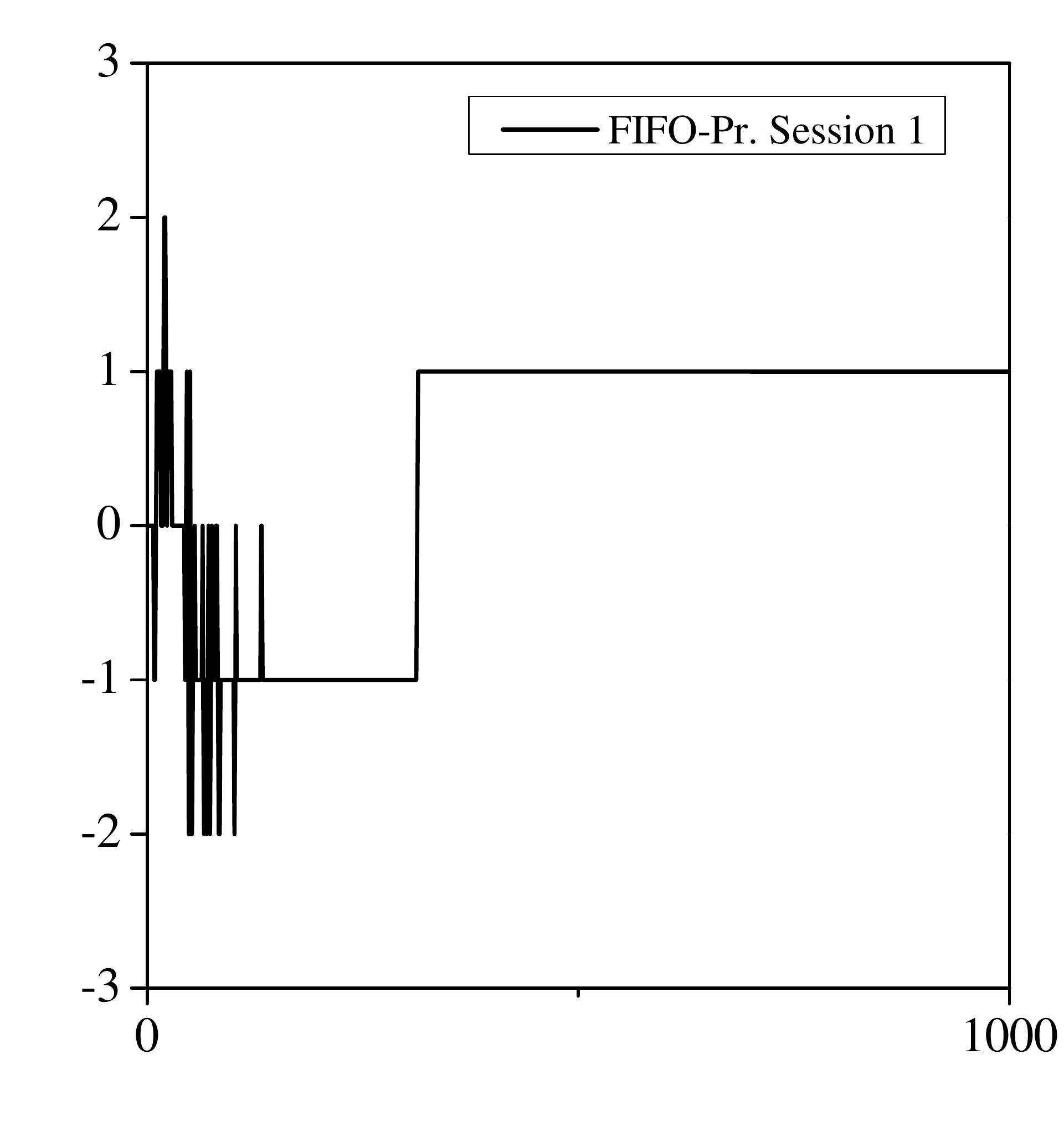}
\put(40,5){\small time (slots)}
\put(19,96){\small total backlog difference}
\end{overpic}
   \caption{Sample path difference in total system backlog, between different  underlay forwarding policies:  (left) difference between FIFO and HLPPS, (right) difference between FIFO and Strict Priority to session 1.}
   \vspace{-0.15in}
     \label{fig:fwddiscipline}
\end{center}
\end{figure}

\begin{table}[t!]
\centering
\begin{tabular}{c | c c c c }
$\lambda$ & FIFO & HLPPS & LQF & Priority Session 1 \\ [0.5 ex]
\hline
$0.8$ & 7.523 & 7.517 & 7.522 & 7.534\\
$0.85$ & 9.529 & 9.505 & 9.529 & 9.541 \\
$0.9$ & 13.240 & 13.245 & 13.193 & 13.238 \\
$0.95$ & 23.850 & 23.887 & 23.899 & 23.893 \\
$0.99$ & 98.738 & 98.605 & 98.755 & 98.624 \\
\hline
\end{tabular}
\caption{Average delay performance of $\TB$ under different   underlay forwarding policies.}
\vspace{-0.65in}
\label{delay_table}
\end{table}

\subsection{Delay Comparison}
 We simulate the delay of different routing policies, comparing the performance of $\TB$ and  $\BPr$ overlay policies, as well as $\BPa$ and $\BPs$ which are not admissible in the overlay. We  experiment 
  for $\lambda_1=\lambda_2=\lambda/2$, and we plot the average total backlogs in the system for two example networks shown to the left of each plot.
 
 In Fig.~\ref{fig:delay1} $\BPr$ fails to detect congestion in the tunnel \textsf{ac} and consequently delay increases for $\lambda>0.7$. 
 We observe that $\TB$ outperforms $\BPa$ and $\BPr$, and performs similarly to $\BPs$.
 This relates to avoidance of cycles at low loads by use of shortest paths, see
 \cite{georgiadis}. In particular, $\BPs$ achieves this by means of hop count bias, while $\TB$ using the tunnels. 
 \emph{A remarkable fact is that  $\TB$ applies control only at the overlay nodes and outperforms in terms of delay $\BPa$ which controls all physical nodes in the network.} 
 
  In Fig.~\ref{fig:delay2} we study queues in tandem, in which case all policies have maximum throughput since there is a unique path through which all the packets travel. We choose this scenario to  demonstrate  another reason why $\TB$ has good delay performance.
  The delay of backpressure increases quadratically to the number of network nodes because of maintaining equal backlog differences across all neighbors \cite{Bui2009}.
   In the case of $\TB$, as well as any other admissible overlay policy like $\BPr$, the backlogs increase with  the number of routers. Thus, when $|{\cal V}|<|\cal N|$ we obtain a delay gain by applying control only at routers.  Fig.~\ref{fig:delay2}  showcases exactly this delay gain that  $\TB$ and $\BPr$ have versus $\BPa$ and $\BPs$.
 
 We conclude that $\TB$ has very good delay performance which is attributed to two main reasons:
\begin{enumerate}
\item When traffic load is low, the majority of the packets follow  shortest paths. The number of packets going in cycles is significantly reduced.
\item Since there is no need for congestion feedback within the tunnels, the backlog buildup is not proportional to the number of network nodes but  to the number of routers.
\end{enumerate}

\begin{figure}[t!]
\begin{center}
\begin{overpic}[scale=.35]{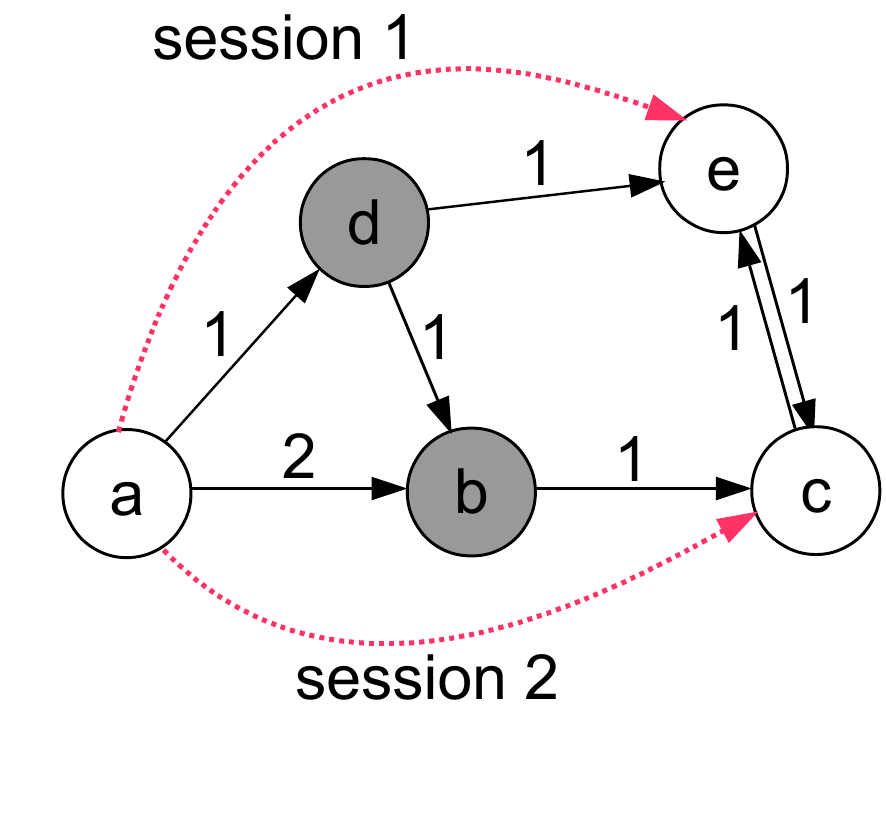}\end{overpic}~
\begin{overpic}[scale=.215]{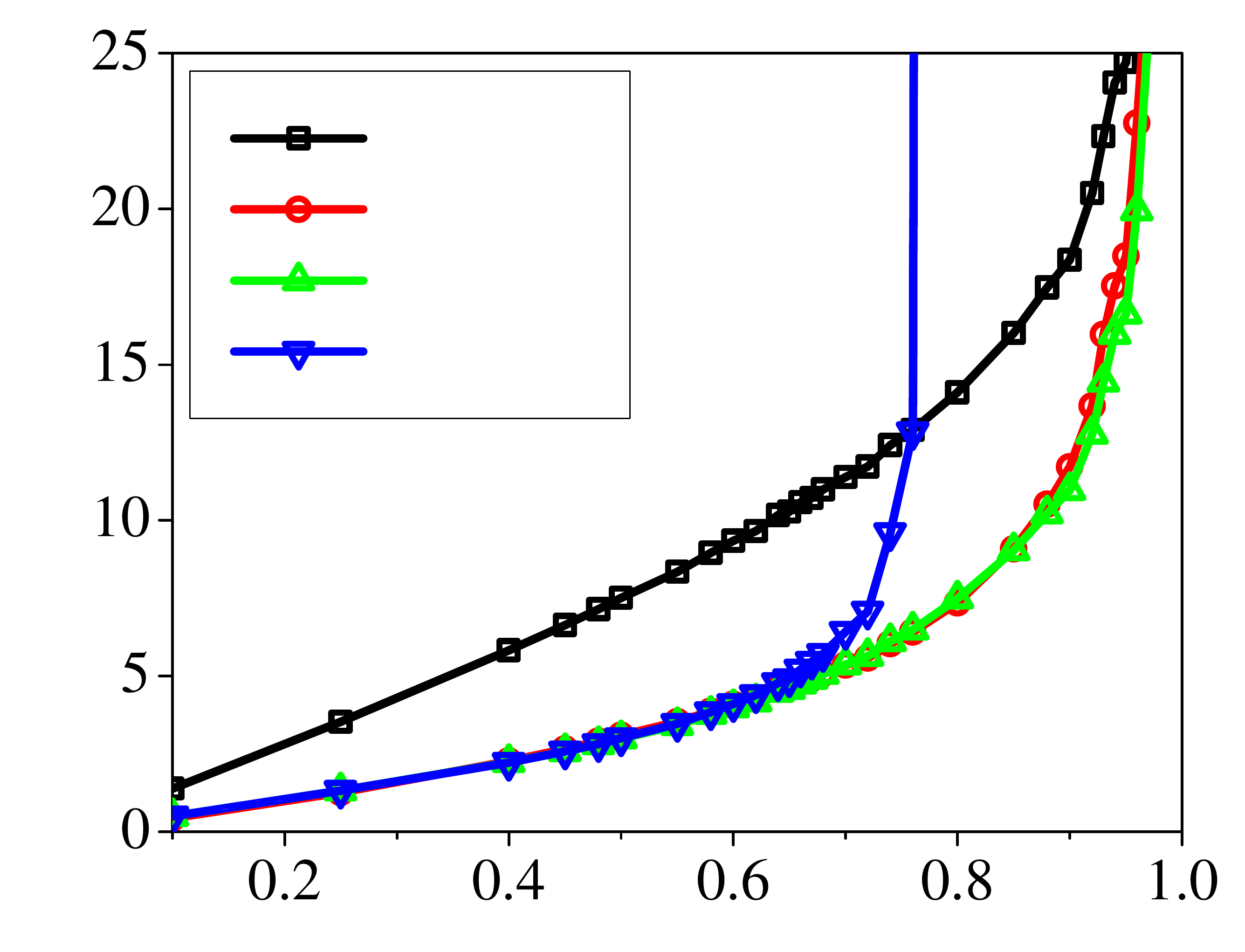}
\put(31,64){\footnotesize $\BPa$}
\put(31,58.5){\footnotesize $\BPs$}
\put(31,52.5){\footnotesize $\TB2$}
\put(31,46.5){\footnotesize $\BPr$}
\put(47,-2){\small load $\lambda$}
\put(30,74){\small average total backlog}
\end{overpic}
   \caption{Delay Comparison: (left) Example under study. (right) Average total  backlog per offered load when $\lambda_1=\lambda_2=\lambda/2$.}
      \vspace{-0.1in}
     \label{fig:delay1}
\end{center}
\end{figure}

\begin{figure}[t!]
\begin{center}
\begin{overpic}[scale=.4]{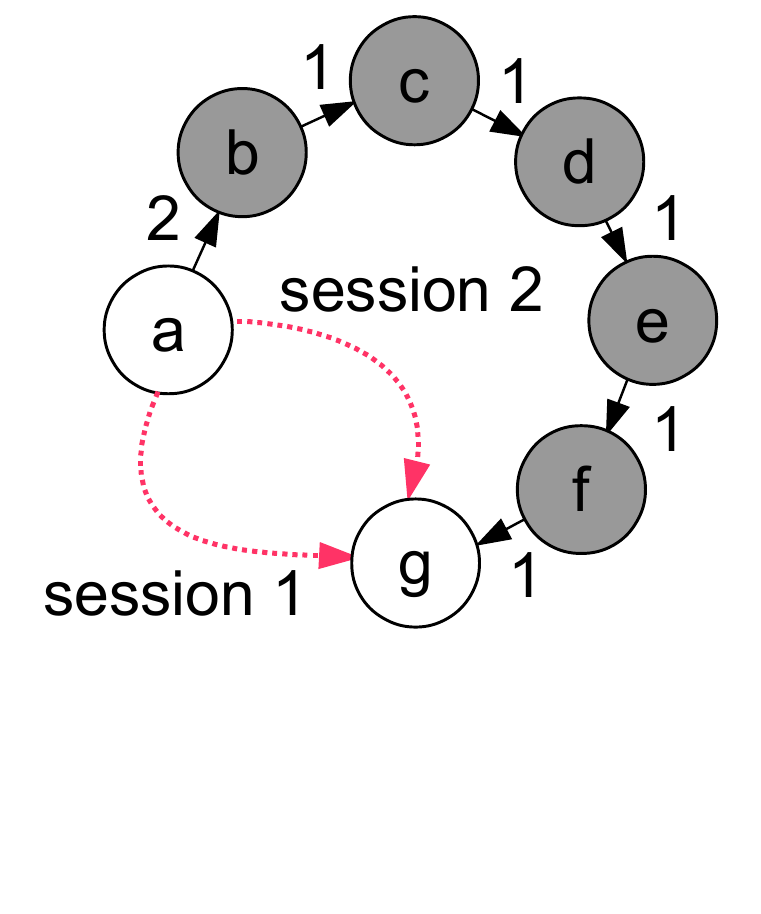}\end{overpic}~~
\begin{overpic}[scale=.215]{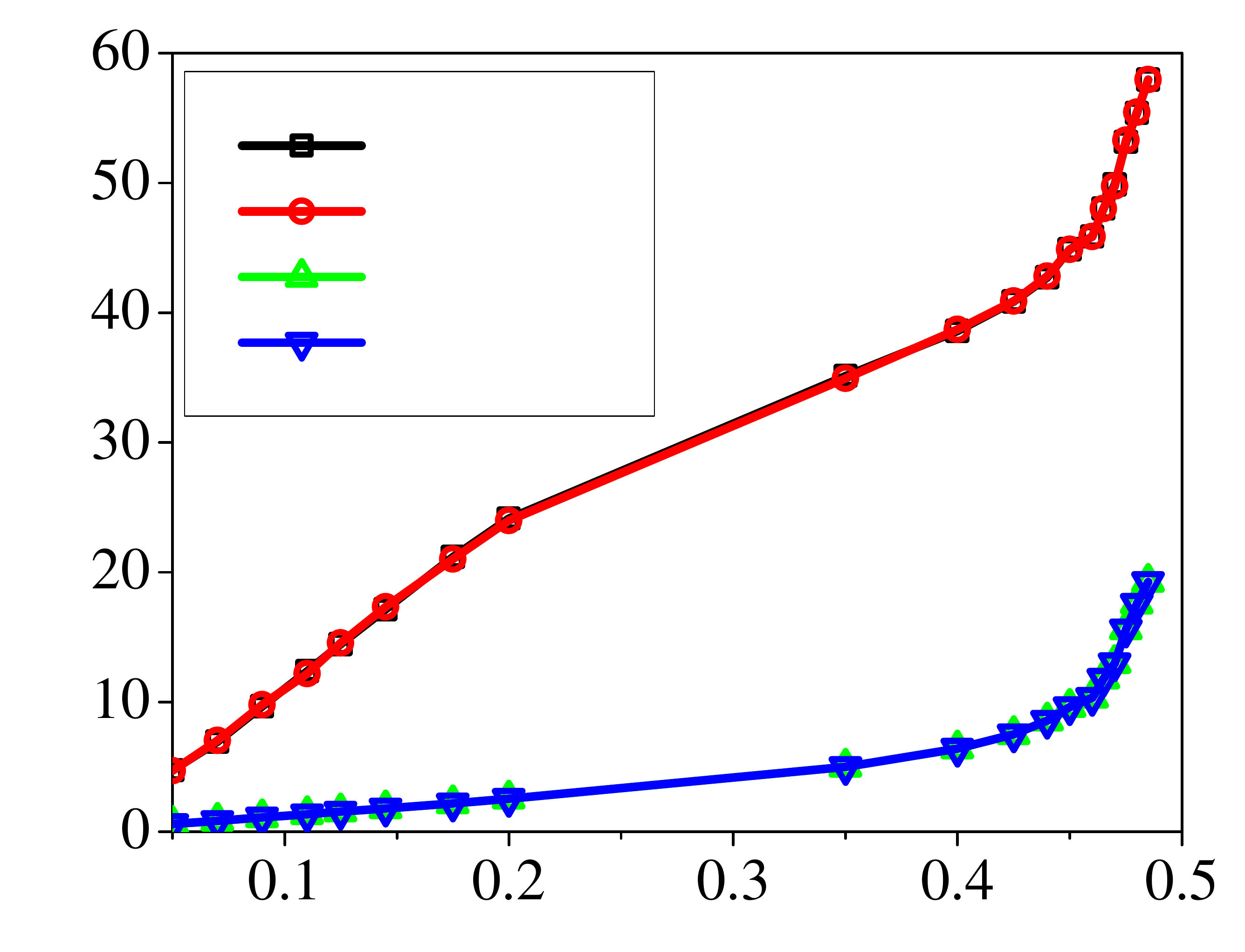}
\put(31,64){\footnotesize $\BPa$}
\put(31,58.5){\footnotesize $\BPs$}
\put(31,52.5){\footnotesize $\TB2$}
\put(31,46.5){\footnotesize $\BPr$}
\put(47,-2){\small load $\lambda$}
\put(30,74){\small average total backlog}
\end{overpic}
   \caption{Delay Comparison: (left) Example under study. (right) Average total  backlog per offered load when $\lambda_1=\lambda_2=\lambda/2$.}
   \vspace{-0.55in}
     \label{fig:delay2}
\end{center}
\end{figure}

\subsection{Applying our Policy to Overlapping Tunnels}

Next we extend $\TB$ to   networks with overlapping tunnels, see the example in  Fig.~\ref{fig:overlapping} (left). 
In this context  Theorem \ref{th:opti} does not apply and we have no guarantees that $\TB$ is maximally stable. 
The key to achieving maximum throughput is to  correctly balance the ratio of  traffic from each session injected into the overlapping tunnels. 
For the network to be stable with load $(.9,.9)$, a policy needs to direct
 most of the traffic of session 1 through the dedicated link $(\textsf{a},\textsf{e})$, or equivalently 
  to allocate $\mu_{\textsf{ac}}^1(t)=0$.
Since node $\textsf{e}$ is the destination of session 1, and hence $Q_{\textsf{e}}^1(t)=0$, 
we need to relate this routing  decision to the congestion in the tunnel. 

To make this work, we introduce the following extension. 
Instead of conditioning transmissions on router differential backlog $Q_i^{c_{ij}^*}(t)>Q_j^{c_{ij}^*}(t)$ as in $\TB$, we use the condition $Q_i^{c_{ij}^*}(t)>Q_j^{c_{ij}^*}(t)+F_{ij}(t)$. 
Intuitively, we expect a non-congested node  to have a small backlog and thus  avoid sending packets over a congested tunnel.
The new policy is called  $\TB$2. 
It can be proven  that $\TB$2 is maximally stable for  non-overlapping tunnels.
Although we do not have a proof for the case of overlapping tunnels, the simulation results show  that by choosing $T$ to be large $\TB$2 achieves maximum throughput.

\noindent \rule[0.05in]{3.5in}{0.01in} 

\vspace{-0.03in}
\begin{centering} \textbf{$\TB$2 for Overlapping Tunnels }

\vspace{-0.03in}
\end{centering}
\noindent \rule[0.05in]{3.5in}{0.01in}

Fix a $T$ to satisfy eq. (\ref{eq:thres}), and recall condition (\ref{eq:C1}):
\begin{equation*}\label{eq:C1b}
F_{ij}(t)<T.
\end{equation*}

In  slot $t$ for tunnel $(i,j)$
 let 
\[
c_{ij}^*\in \argmax_{c\in{\cal C}} Q_i^c(t)-Q_j^c(t),
\]
be a session that maximizes the differential backlog between  router $i,j$, ties resolved arbitrarily. 
 Then route into tunnel $(i,j)$ 
\begin{equation}\label{eq:servf2}
\mu_{ij}^{c_{ij}^*}(t,\text{TB})=\left\{\begin{array}{ll}
\Cin & \text{ if  } Q_i^{c_{ij}^*}(t)>Q_j^{c_{ij}^*}(t)+F_{ij}(t) \\
& \text{ AND } (\ref{eq:C1}) \text{ is true}\\
& \\
0 & \text{ otherwise} 
\end{array}\right.
\end{equation}
and $\mu_{ij}^{c}(t,\text{\TB})=0,~\forall c\neq c_{ij}^*$. Recall, that $\Cin$ denotes the  capacity of physical link  that connects router $i$ to the tunnel $(i,j)$. 

\noindent \rule[0.05in]{3.5in}{0.01in}

Figure \ref{fig:overlapping} shows the results from an experiment where $T=10$
, $\lambda_1=\lambda_2=\lambda$,  and we vary $\lambda$. $\TB2$ achieves full throughput and similar delay to $\BPs$, doing strictly better than $\BPr,\BPa$. To understand how $\TB2$ works, consider the sample path evolution (Fig.~\ref{fig:overlapping2}), where $Q_{\textsf{a}}^1(t)-Q_{\textsf{e}}^1(t),Q_{\textsf{b}}^2(t)-Q_{\textsf{f}}^2(t),F_{\textsf{ae}}(t)$ are shown. 
Most of the time we have $Q_{\textsf{a}}^1(t)-Q_{\textsf{e}}^1(t)<10$, thus
by the choice of $T=10$  and the condition used in (\ref{eq:servf2}),  session 1 rarely gets the opportunity to transmit packets to the overlapping tunnels. 
As $T$ increases session 1 will get  fewer and fewer  opportunities, hence $\TB2$ behavior will approximate the optimal.
In Fig~\ref{fig:overlapping2} (right) we plot the average total backlog for different values of $T$.
As $T$ increases, the performance at high loads improves.

\begin{figure}[t!]
\begin{center}
\begin{overpic}[scale=.39]{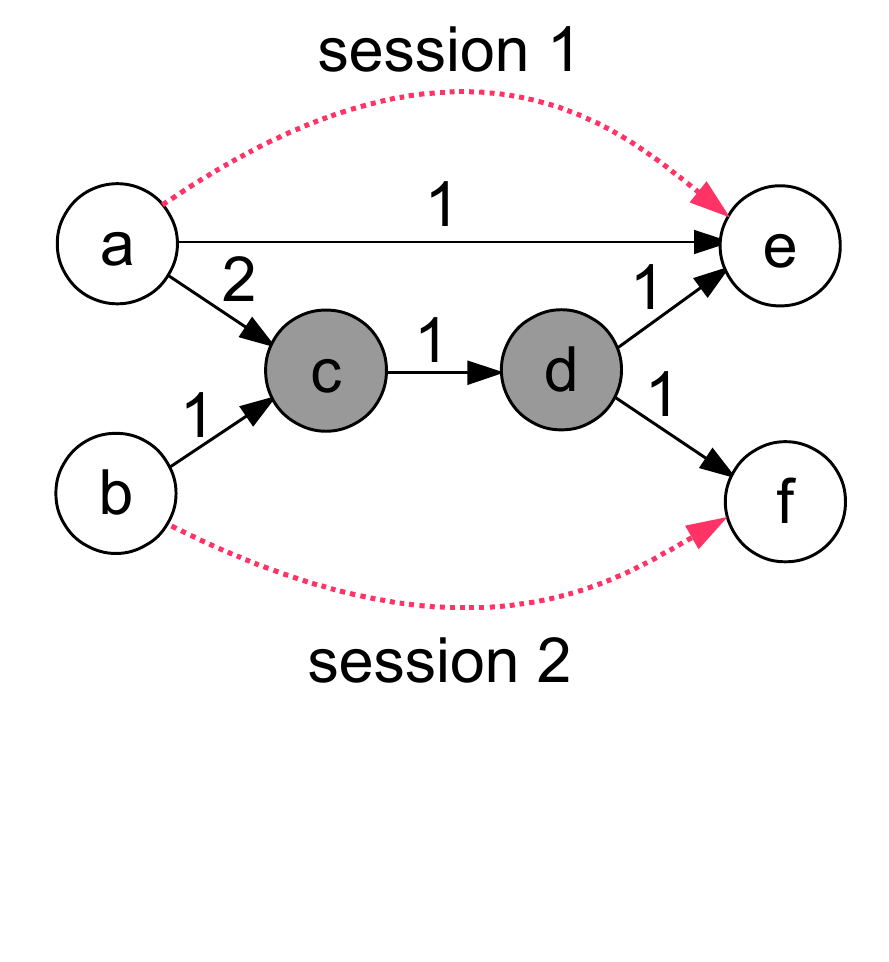}\end{overpic}~\hspace{-0.2in}
\begin{overpic}[scale=.215]{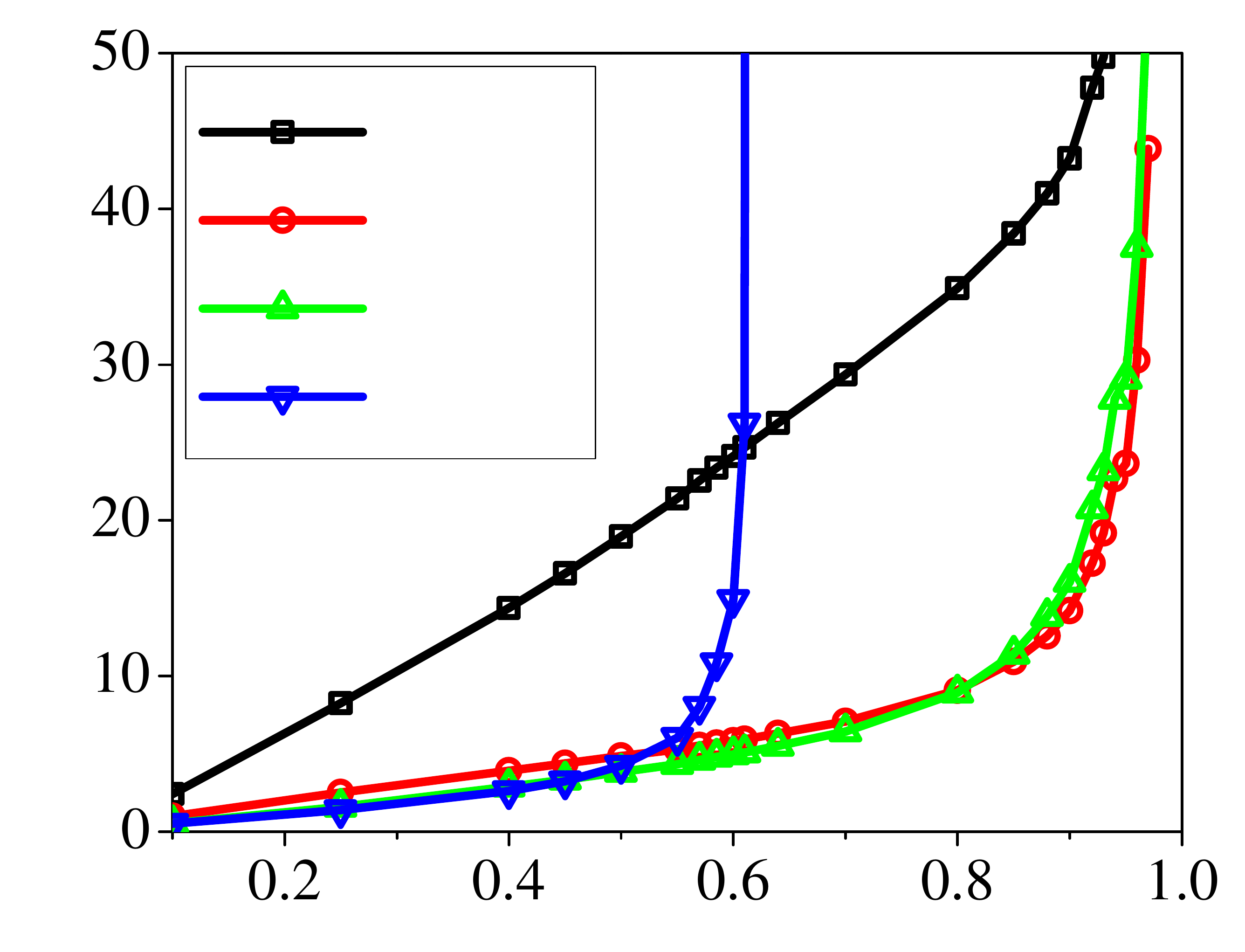}
\put(31,65){\footnotesize $\BPa$}
\put(31,44){\footnotesize $\BPr$}
\put(31,51){\footnotesize $\TB2$}
\put(31,58){\footnotesize $\BPs$}
\put(47,-2){\small load $\lambda$}
\put(30,74){\small average total backlog}
\end{overpic}
   \caption{Overlapping Tunnels: (left) Example under study. (right) Average total  backlog per offered load when $\lambda_1=\lambda_2=\lambda/2$.}
      \vspace{-0.1in}
     \label{fig:overlapping}
\end{center}
\end{figure}

\begin{figure}[t!]
\begin{center}
\hspace{-0.14in}\begin{overpic}[scale=.17]{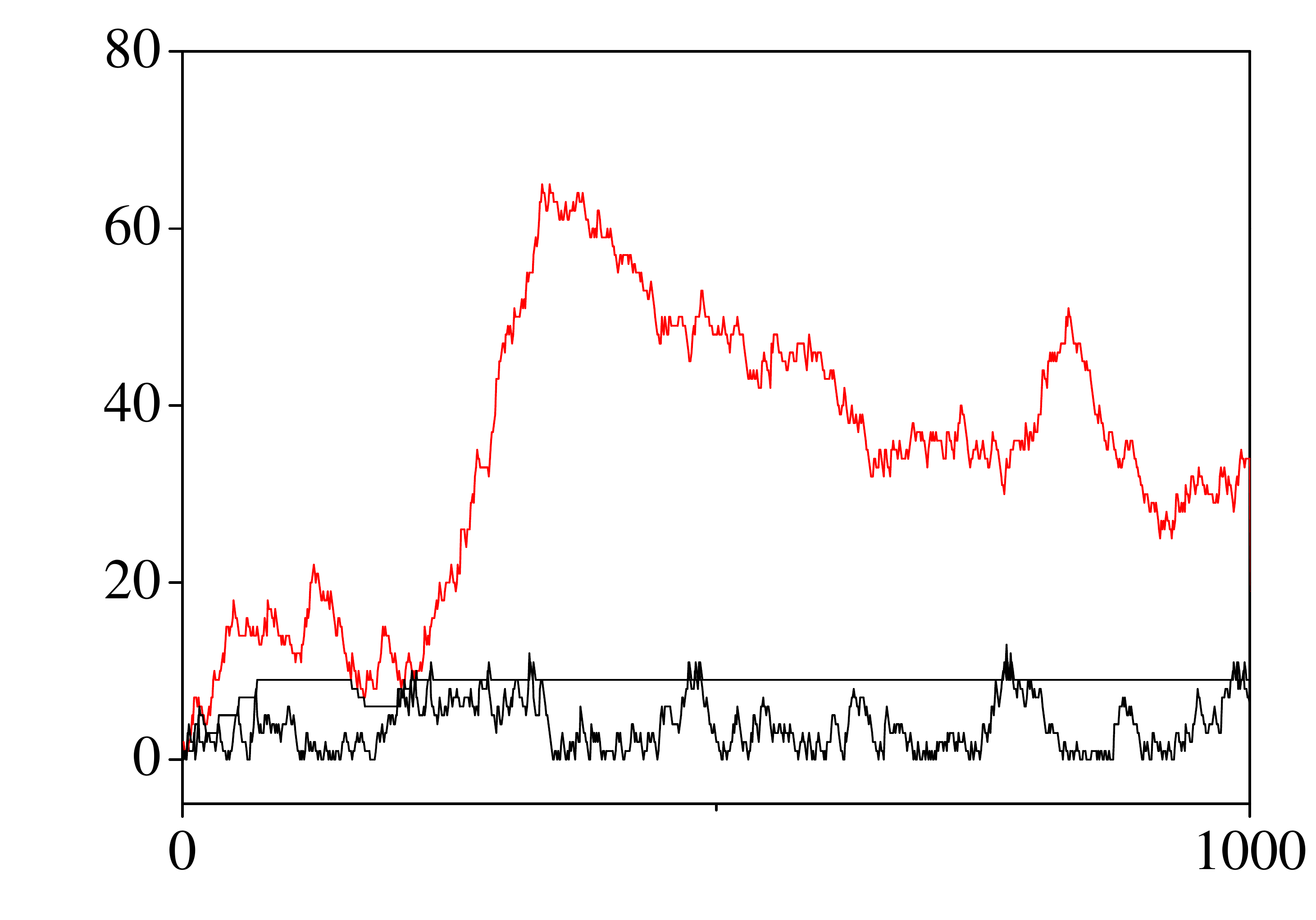}
\put(58,25){\footnotesize $Q_{\textsf{a}}^1(t)-Q_{\textsf{e}}^1(t)$}
\put(57,55){\footnotesize $Q_{\textsf{b}}^2(t)-Q_{\textsf{f}}^2(t)$}
\put(40,32){\footnotesize $F_{\textsf{a}\textsf{e}}(t)$}
\put(40,2){\small time (slots)}
\put(30,67.6){\small individual backlogs}
\put(45,30){\color{black}\vector(0,-1){11.5}}
\put(59,53){\color{black}\vector(-1,-3){2.5}}
\put(62,22){\color{black}\vector(-1,-2){3.5}}
\end{overpic}~\hspace{-0.2in}
\begin{overpic}[scale=.17]{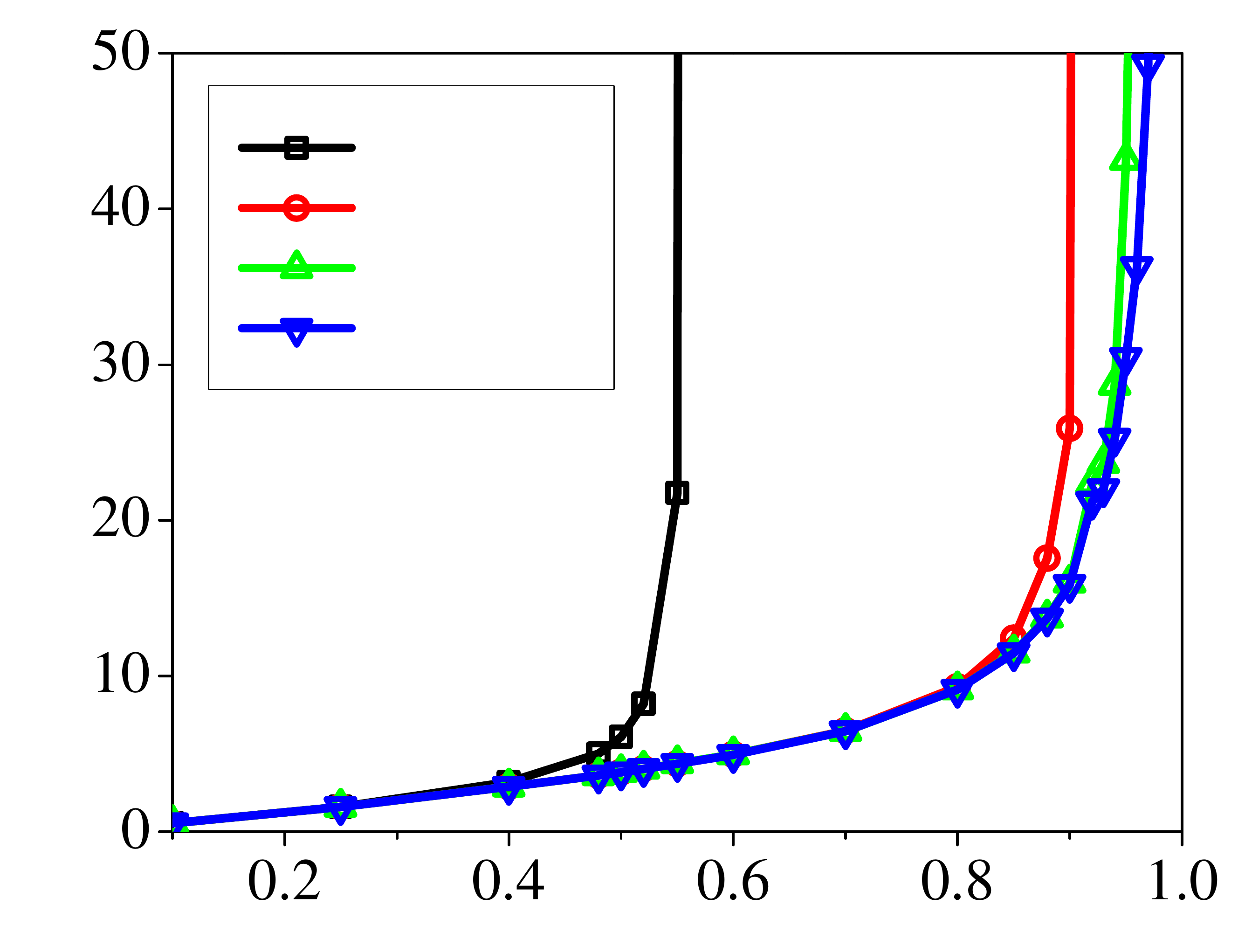}
\put(31,63.5){\footnotesize $T$=2}
\put(31,58){\footnotesize $T$=5}
\put(31,53){\footnotesize $T$=10}
\put(31,48){\footnotesize $T$=25}
\put(47,-2){\small load $\lambda$}
\put(27,74){\small average total backlog}
\end{overpic}
   \caption{(left) System evolution (one sample path) for $\lambda_1=\lambda_2=.97$, $T=10$. (right) Average total  backlog per offered load when $\lambda_1=\lambda_2=\lambda/2$.}
   \vspace{-0.3in}
     \label{fig:overlapping2}
\end{center}
\end{figure}
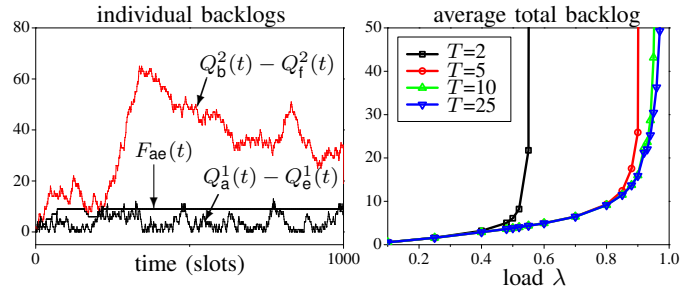

\section{Conclusions}

In this paper we propose a backpressure extension which can be applied in overlay networks. 
From prior work, we know that if the overlay is designed wisely, it can match the throughput of the physical network \cite{C_jones_14}.
Our contribution is to prove that  the maximum overlay throughput can be achieved by means of dynamic routing.
Moreover, we show that our proposed scheme $\TB$ makes the best of both worlds (a) efficiently choosing the paths in online fashion adapting to network variability and (b) keeping average delay small avoiding the known inefficiencies of the legacy backpressure scheme.

Future work involves the mathematical analysis of the overlapping tunnels case and the consideration of wireless transmissions. In both cases Lemma \ref{lem:leaky} does not hold due to correlation of routing decisions at routers with scheduling at forwarders. 

\section{Acknowledgments}

We would like to thank Dr.~Chih-Ping Li and Mr.~Matthew Johnston for their helpful discussions and comments.


\appendices
%
%
%
\section{Proof of Lemma~\ref{lem:leaky}}\label{app:lem}
\newtheorem*{lem:1}{Lemma \ref{lem:leaky}}
\begin{lem:1}[Output of a Loaded Tunnel]
Under any control policy $\pi\in\Pi$,  suppose that in time slot $t$ the total tunnel backlog satisfies $ F_{ij}(t)>T_0$, for some $(i,j)\in{\cal E}$, where $T_0$ is defined in \eqref{eq:T0}. 
The instantaneous output of the tunnel satisfies 
\begin{equation}\label{eq:leaky}
\sum_c\phi_{ij}^c(t)=\pacap.
\end{equation}
\end{lem:1}
\begin{proof}[Proof of Lemma \ref{lem:leaky}]
Consider  a  tunnel $(i,j)$ which forwards packets, using  an arbitrary work-conserving policy, over the path $p_{ij}$ with $M_{ij}$ underlay nodes.
Renumber the nodes in the path in sequence they are visited by packets as $0,1,\dots,M_{ij}+1$, where $0$ refers to $i$ and $M_{ij}+1$ to $j$, hence
\[
p_{ij}\triangleq \{0,1,\dots,M_{ij},M_{ij}+1\}.
\]
Since the statement  is inherently related to packet forwarding internally in the tunnel $(i,j)$, we will introduce some notation.
Denote by $F_{ij}^{k}(t), k=1,\dots, M_{ij}$ the packets waiting at the $k^{\text{th}}$ node  at slot $t$, 
to be transmitted to the $k+1^{\text{th}}$, along tunnel $(i,j)\in {\cal V}$ (the packets may belong to different sessions). 
Clearly, it is $\sum_{k=1}^{M_{ij}}F_{ij}^{k}(t)=F_{ij}(t)$.
Also, let $\phi_{ij}^{k,c}(t)$ be the actual number of session $c$ packets that leave this node in slot $t$. 
For all $(i,j), k, c, t$, due to work-conservation we have
\begin{equation}\label{eq:Fc1}
\sum_c\phi_{ij}^{k,c}(t)=\min\{R_{k}, F_{ij}^{k}(t)\},
\end{equation}
$R_k$ denoting the capacity of the physical link connecting nodes $k,k+1$.
Hence,  $F_{ij}^{k}(t), k=1,\dots ,M_{ij}$ evolve as
\begin{equation}\label{eq:indf}
F_{ij}^{k}(t+1)=F_{ij}^{k}(t)-\sum_c\phi_{ij}^{k,c}(t)+\sum_c\phi_{ij}^{k-1,c}(t).
\end{equation}

First we establish  that the instantaneous output of the tunnel cannot be larger than its bottleneck capacity, i.e.,   \begin{equation}\label{eq:req}\sum_c\phi_{ij}^c(t)\leq\pacap.\end{equation} 
If the bottleneck link is the last link 
on $p_{ij}$ then (\ref{eq:req}) follows immediately from \eqref{eq:Fc1}. Else, pick $k$ such that  $0\leq k<M_{ij}$ and suppose $(k,k+1)$ is the bottleneck link. 
 Then let us focus on the link $(k+1,k+2)$.
 For its input we have 
\[\sum_c\phi_{ij}^{k,c}(t)\stackrel{\eqref{eq:Fc1}}{\leq} R_{k}\triangleq R_{ij}^{\min},\quad \text{for all}~t\]
where above and in the remaining proofs we use parentheses to denote the expressions from which equalities and inequalities follow.
For link $(k+1,k+2)$  output
\[\sum_c\phi_{ij}^{k+1,c}(t)= \min\{F_{ij}^{k+1}(t), R_{k+1}\},\]
where $R_{k+1}\geq R_{k}$. Starting the system empty, the backlog $F_{ij}^{k+1}(t)$ cannot grow larger than $R_{k}$ since this is the maximum number of arriving packets in one slot and they are all served in the next slot. Hence, it is also $\sum_c\phi_{ij}^{k+1,c}(t)=F_{ij}^{k+1}(t)\leq R_{k}$.
By induction, the same is true for  $F_{ij}^{l}(t),\phi_{ij}^{l}(t)$ for any $k<l\leq M_{ij}$, 
and we get
 (\ref{eq:req}).

The remaining proof is by contradiction. Assume $\sum_c\phi_{ij}^c(t)<\pacap$.
Consider the physical link $(k,k+1)$ with $k=2,\dots,M_{ij}$. Using \eqref{eq:indf} 
\begin{equation}\label{eq:leakyback}
F_{ij}^{k}(t)<\pacap ~\Rightarrow~ F_{ij}^{k-1}(t-1)<\pacap.
\end{equation}
To understand (\ref{eq:leakyback}) note that if the RHS was false, by (\ref{eq:Fc1}) we would have $\sum_c\phi_{ij}^{k-1,c}(t-1)\geq \pacap$ and thus by (\ref{eq:indf}) also  $F_{ij}^{k}(t)\geq\pacap$.

\begin{figure}[t!]
\begin{center}
\begin{overpic}[scale=.315]{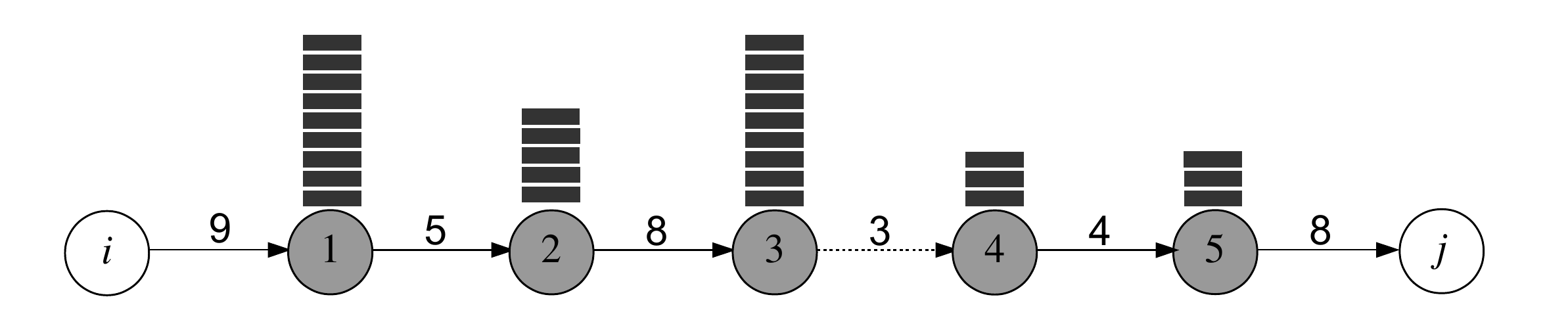}
\end{overpic}
   \caption{An overloaded tunnel with bottleneck capacity $\pacap=3$.}
   \vspace{-0.4in}
     \label{fig:leaky}
\end{center}
\end{figure}

Since by the premise we have $\sum_c\phi_{ij}^{M_{ij},c}(t)\equiv\sum_c\phi_{ij}^c(t)<\pacap$
, applying   (\ref{eq:Fc1}) we deduce $F_{ij}^{M_{ij}}(t)<\pacap$ from which applying (\ref{eq:leakyback}) recursively we roll back in time and space to obtain 
\[F_{ij}^{k}(t-M_{ij}+k)<\pacap,~~k=1,\dots,M_{ij}.\]
Since the maximum backlog increase at any node within one slot is $R_{ij}^{\max}$, we roll forward in time to get
\[F_{ij}^{k}(t)<\pacap+(M_{ij}-k)R_{ij}^{\max},~~k=1,\dots,M_{ij}.\]
Summing up for all forwarders $k=1,\dots,M_{ij}$ we get 
\begin{align}\label{eq:contra}
F_{ij}(t)&=\sum_{k=1}^{M_{ij}}F_{ij}^{k}(t)<\sum_{k=1}^{M_{ij}} \left[\pacap+(M_{ij}-k)R_{ij}^{\max}\right]\notag\\
&=M_{ij}\pacap+\frac{M_{ij}(M_{ij}-1)}2R_{ij}^{\max}\stackrel{\eqref{eq:T0}}{=}T_0.
\end{align}
which contradicts the premise of the lemma.
\end{proof}

\section{Proof of Theorem \ref{th:opti}}\label{app:th}
\begin{proof}[Proof of Theorem~\ref{th:opti}]
In order to prove that $\TB$ is maximally stable, we will pick an arbitrary arrival vector $\boldsymbol \lambda$ in the interior of $\Lambda({\cal V})$ and show that the system is stable.
To prove stability  we perform a $K$-slot drift analysis and show that $\TB$ has a negative drift. 
Our system state is described by the  vector of queue lengths $\mat H_t\triangleq \left( (Q_i^c(t)),(F_{ij}(t)) \right)$. By Lemma \ref{lem:detFb},  the tunnel backlogs $(F_{ij}(t))$ are deterministically bounded under  $\TB$, and thus for the purposes of showing $\TB$ stability we choose  the  candidate quadratic  Lyapunov function:
\begin{equation}\label{eq:lyapunov}
L(\mathbf H_t)\triangleq\frac{1}2 \sum_i \left[Q_i^c(t)\right]^2.
\end{equation}
We will use the following shorthand notation
\[
\mathbbm{E}_\mathbf H\{.\}\equiv \mathbbm{E}\left\{.|\mathbf H_t, F_{ij}(t)\leq F^{\max}, \forall (i,j)\right\}.
\]
The $K$-slot Lyapunov drift under  policy $\pi$ is 
\[
\Delta^{\pi}_K(t)\triangleq  \mathbbm{E}\{L(\mathbf H_{t+K})-L(\mathbf H_t) |\mathbf H_t\}.
\]
From Lemma \ref{lem:detFb} we have  $F_{ij}(t)\leq F^{\max}$ for every sample path,  and thus the $K$-slot Lyapunov drift for TB becomes  $\driftK=\mathbbm{E}_{\mathbf H}\{L(\mathbf H_{t+K})-L(\mathbf H_t)\}$.
To prove the stability of $\TB$, it suffices to show that for any $\boldsymbol\lambda$ in the interior of the stability region there exist positive constants $\eta,\xi$ and a  finite $K$ such that 
$\driftK\leq \eta- \xi\sum_{i,c}Q_i^c(t)$, see $K$-slot drift theorem in \cite{georgiadis} (corollary of the Foster's criterion).
The remaining proof shows this fact.

To derive an expression for the $K$-slot drift $\driftK$ 
we first write the $K$-slot queue evolution inequalities 
\begin{align}
& Q_i^c(t+K)\leq \left(Q_i^c(t)-\sum_{b\in\mathcal V} \tilde\mu_{ib}^{c}(t,\pi)\right)^+\hspace{-0.12in}+\hspace{-0.05in}\sum_{a\in\mathcal V} \tilde\phi_{ai}^{c}(t)+\tilde A_i^c(t),\label{eq:kFeq}\\
& F_{ij}(t+K)\leq F_{ij}(t)-\sum_c\tilde\phi_{ij}^{c}(t)+\sum_c\tilde\mu_{ij}^{c}(t,\pi), \label{eq:kFeq2}
\end{align}
where use the $\tilde{(.)}$ notation to denote summations over $K$ slots: 
\begin{align*}
&\tilde A_i^c(t)\triangleq\sum_{\tau=0}^{K-1} A_i^c(t+\tau),\\
&\tilde\mu_{ij}^{c}(t,\pi) \triangleq \sum_{\tau=0}^{K-1} \mu_{ij}^{c}(t+\tau,\pi),\\ 
&\tilde\phi_{ij}^{c}(t) \triangleq \sum_{\tau=0}^{K-1} \phi_{ij}^{c}(t+\tau).
\end{align*}
The inequality  \eqref{eq:kFeq} is because  the arrivals $\sum_{a\in\mathcal V} \tilde\phi_{ai}^{c}(t)+\tilde A_i^c(t)$ are added at the end of the $K$-slot period---some of these packets may actually be served within the $K$-slot period. 

Taking squares on (\ref{eq:kFeq}), using Lemma 4.3 from \cite{georgiadis}, and performing some calculus we obtain the following bound
\begin{align*}
\driftK \leq  K^2B_1&+\sum_{c,i}K\lambda_i^c Q_i^c(t)\\
&\hspace{-0.5in}-\sum_{c,i} Q_i^c(t) \mathbbm{E}_{\mathbf H}\left\{ \sum_b\tilde\mu_{ib}^{c}(t,\text{\TB})-\sum_a\tilde\phi_{ai}^c(t) \right\}.
\end{align*}
where $B_1\triangleq d_{\max}^2R_{\max}^2+ A_{\max}^2/2+A_{\max}d_{\max}R_{\max}$ is a positive constant related to the maximum number of arriving packets in a slot $A_{\max}$, the maximum link capacity $R_{\max}$, and the maximum node-degree $d_{\max}$ in graph ${\cal G}_R$.

Denote with $X^c_{ij}(t)$ the session $c$ packets in the tunnel $(i,j)$, where $\sum_cX^c_{ij}(t)=F_{ij}(t)$. 

This backlog evolves as
\[
X^c_{ij}(t+K)\leq X^c_{ij}(t)-\tilde\phi_{ij}^c(t)+\tilde\mu_{ij}^c(t).
\]
We have $X^c_{ij}(t+K)\geq 0$, and $X^c_{ij}(t)\leq F_{ij}(t)\leq F^{\max}$, hence
  $X_{ij}^c(t+K)-X^c_{ij}(t)\geq -F^{\max}$. It follows that for any $t,K$   \[\sum_a\tilde\phi_{ai}^c(t)\leq \sum_a\tilde\mu_{ai}^{c}(t,\text{\TB})+d_{\max}F^{\max},\] where  $F^{\max}$ is the deterministic upper bound of $F_{ij}(t)$ from (\ref{eq:Fmax}).
 Hence,

\begin{align}
&\Delta^{\text{\TB}}_K(t)- K^2B_1-\sum_{c,i}(K\lambda_i^c+d_{\max}F^{\max}) Q_i^c(t)\notag\\
&\leq-\sum_{c,i} Q_i^c(t) \mathbbm{E}_{\mathbf H}\left\{ \sum_b\tilde\mu_{ib}^{c}(t,\text{\TB})-\sum_a\tilde\mu_{ai}^{c}(t,\text{\TB}) \right\}\notag\\
&= - \sum_{c,(i,j)} \mathbbm{E}_{\mathbf H}\left\{\tilde\mu_{ij}^{c}(t,\text{\TB})\left[Q_i^c(t)-Q_{j}^c(t)\right] \right\},\label{eq:interim}
\end{align}
where the equality comes from the node-centric and link-centric packet accounting in a network, see \cite{georgiadis} on page 48.

\subsection{An Oracle Policy}
\label{sec:stdpol}

We design a stationary oracle ($\OR$) policy, whose purpose is to assist us in proving the optimality of $\TB$ policy. 
The foundation of $\OR$ lies on the existence of  a flow decomposition.
For any $\boldsymbol\lambda$ in the interior of the stability region, there exists an $\epsilon$ such that $\boldsymbol\lambda^{\epsilon}\equiv\boldsymbol\lambda+\epsilon\mathbf{1}$ is also stabilizable, where $\mathbf{1}$ is a vector of ones. Thus, by the sufficiency of the conditions in section \ref{sec:region} there must exist a feasible flow decomposition $(f_{ij}^{c,\boldsymbol\lambda^{\epsilon}})$ such that 
\[
\sum_{a\in\mathcal V} f_{ai}^{c,\boldsymbol\lambda^{\epsilon}}- \sum_{b\in \mathcal V}f_{ib}^{c,\boldsymbol\lambda^{\epsilon}}\geq \lambda_i^c+\epsilon , \text{ for all } i\in\mathcal{V}
\]
and $\sum_cf_{ij}^{c,\boldsymbol\lambda^{\epsilon}}< \pacap$ for all $(i,j)\in \mathcal{V}$.  Using this particular decomposition we define a specific $\OR$ policy for the particular $\boldsymbol \lambda$ as follows.

\noindent \rule[0.05in]{3.5in}{0.01in} 

\vspace{-0.03in}
\begin{centering} $\boldsymbol \lambda$--\textbf{Stationary Randomized ORacle ($\OR$) Policy }

\vspace{-0.03in}
\end{centering}
\noindent \rule[0.05in]{3.5in}{0.01in}
%

In every time slot and at each tunnel $(i,j)$, 
\begin{itemize}
\item if $F_{ij}(t)\geq T$ (the tunnel is loaded), then choose 
\begin{equation}\label{eq:zeroallocation}
\mu_{ij}^{c}(t,\OR)=0,~~\forall c\in\mathcal{C},
\end{equation}

\item else if $F_{ij}(t)< T$ (not loaded  tunnel),  choose a session using an i.i.d. process $N(t)$ with  distribution 
\[
P\left(N(t)=c'\right)=\frac{f_{ij}^{c',\boldsymbol\lambda^{\epsilon}}}{\sum_cf_{ij}^{c,\boldsymbol\lambda^{\epsilon}}},~~c'=1,\dots,|\mathcal{C}|.
\]
The routing functions are then determined by 
\begin{equation}\label{eq:allocation}
\mu_{ij}^{N(t)}(t,\OR)=\left\{\begin{array}{ll}
\pacap & \text{with prob.}~\frac{\sum_cf_{ij}^{c,\boldsymbol\lambda^{\epsilon}}}{\pacap} \\
0 & \text{with prob.}~1-\frac{\sum_cf_{ij}^{c,\boldsymbol\lambda^{\epsilon}}}{\pacap} 
\end{array}
\right.
\end{equation}
and $\mu_{ij}^{c}(t,\OR)=0,~~\forall c\neq N(t)$.
\footnote{We remark that $N(t)$ and the allocation of service to session $N(t)$ given by (\ref{eq:allocation}) are independent.}\end{itemize}
\noindent \rule[0.05in]{3.5in}{0.01in}

Observe that  $\OR$ satisfies the capacity constraints at every slot, namely $0\leq \sum_c\mu_{ij}^{c}(t,\OR)\leq \Cin$.
Therefore  $\OR\in\Pi$.
Despite wasting transmissions when the tunnels are loaded, $\OR$ stabilizes $\boldsymbol\lambda$:
%
%
\begin{lemma}[$\OR$ $K$-slot performance]\label{cor:stat}
For any $\boldsymbol\lambda$ in the interior of the stability region we have
\begin{align}\label{eq:domin}
&\mathbbm{E}_{\mathbf H}\left\{ \sum_b\tilde\mu^{c}_{ib}(t,\OR)-\sum_a \tilde\mu_{ai}^{c}(t,\OR)\right\}\\
&\hspace{0.5in}\geq K(\lambda_i^c+\epsilon)-d_{\max}\bound , \text{ for all } i\in\mathcal{V}.\notag
\end{align}
\end{lemma}

$\OR$ is also designed to mimic the condition (\ref{eq:C1}) used by $\TB$. 
Because of  it, we can show that  $\TB$ compares favorably to $\OR$.

\begin{lemma}[$K$-slot comparison $\TB$ vs $\OR$]\label{lem:kslot} 
The $K$-slot policy comparison yields for all $(i,j)\in{\cal E}$
\begin{align}\label{eq:nathan}
&\mathbbm{E}_{\mathbf H}\left\{\sum_{c} \tilde\mu_{ij}^{c}(t,\TB)  \left[ Q_i^c(t)-Q_j^c(t)\right]  \right\} \\
&  \hspace{0.1in} \geq \mathbbm{E}_{\mathbf H}\left\{\sum_{c}\tilde \mu_{ij}^{c}(t,\OR)\left[ Q_i^c(t)-Q_j^c(t)\right] \right\} -K^2B_2,\notag
\end{align}
where $B_2\triangleq R_{\max} \left(2d_{\max}R_{\max}+A_{\max}\right)$ is a constant.
\end{lemma}

\subsection{Completing the Proof}
We combine (\ref{eq:interim}) with Lemma \ref{lem:kslot} to get
\begin{align*}
&\Delta^{\text{\TB}}_K(t)- K^2B_1-\sum_{c,i}(K\lambda_i^c+d_{\max}F^{\max}) Q_i^c(t)\\
&\leq K^2 |{\cal E}|B_2- \sum_{c,(i,j)} \mathbbm{E}_{\mathbf H}\left\{\tilde\mu_{ij}^{c}(t,\OR)\left[Q_i^c(t)-Q_{j}^c(t)\right]  \right\},
\end{align*}
which can be rewritten as 
\begin{align*}
&\Delta^{\text{\TB}}_K(t)- K^2(|{\cal E}|B_2+B_1)-\sum_{c,i}(K\lambda_i^c+d_{\max}F^{\max}) Q_i^c(t)\\
&\leq  -\sum_{c,i} Q_i^c(t)\mathbbm{E}_{\mathbf H}\left\{\sum_b\tilde\mu_{ib}^{c}(t,\OR)-\sum_a\tilde\mu_{ai}^{c}(t,\OR)  \right\}\notag\\
&\leq -\sum_{c,i} Q_i^c(t)\left[K(\lambda_i^c+\epsilon)-d_{\max}\bound\right],\notag
\end{align*}
where in the last inequality we used Lemma \ref{cor:stat}. Hence, we finally get
\begin{equation}\label{eq:lastdrift}
\Delta^{\text{\TB}}_K(t)\leq K^2(|{\cal E}|B_2+B_1)-\sum_{c,i}\left[K\epsilon -2d_{\max}\bound\right] Q_i^c(t)
\end{equation}
Choose a finite $K>\frac{2d_{\max}\bound}{\epsilon}$ and define the positive constants 
$\eta\triangleq K^2(|{\cal E}|B_2+B_1)$ and $\xi\triangleq K\epsilon -2d_{\max}\bound $.
Then rewrite  (\ref{eq:lastdrift})  as
\[
\Delta^{\text{\TB}}_K(t)\leq \eta- \xi\sum_{c,i}Q_i^c(t),
\]
which completes the proof.
\end{proof}
Below we give the proofs for the technical lemmas \ref{cor:stat} and \ref{lem:kslot}.

\section{Proof of Lemma \ref{cor:stat}}
\newtheorem*{lem:4}{Lemma \ref{cor:stat}}
\begin{lem:4}[$\OR$ $K$-slot performance]
For any $\boldsymbol\lambda$ in the interior of the stability region we have
\begin{align}\label{eq:domin}
&\mathbbm{E}_{\mathbf H}\left\{ \sum_b\tilde\mu^{c}_{ib}(t,\OR)-\sum_a \tilde\mu_{ai}^{c}(t,\OR)\right\}\\
&\hspace{0.5in}\geq K(\lambda_i^c+\epsilon)-d_{\max}\bound , \text{ for all } i\in\mathcal{V}.\notag
\end{align}
\end{lem:4}
\begin{proof}[Proof of Lemma~\ref{cor:stat}]
First we will need a technical lemma, which states that  a non-loaded tunnel cannot become loaded under $\OR$. We emphasize that in the following lemma all  backlogs $F_{ij}(t)$ refer to the system evolution under $\OR$.
\begin{lemma}\label{lem:absorption}
Consider the system evolution on router edge $(i,j)$ under $\OR$ for the slots $t,t+1,\dots$ and suppose that $F_{ij}(t)$ is arbitrary. Suppose that for a time slot $\tau_0>t$  we have $F_{ij}(\tau_0)<T$, then 
\[ F_{ij}(\tau)<T, \quad \forall \tau>\tau_0. \]
\end{lemma}
\begin{proof}[Proof of lemma \ref{lem:absorption}]
The proof is by contradiction. Suppose there exists $\tau'$ such that $F_{ij}(\tau')\geq T$ and $\tau'>\tau_0$. Then, there must exist a slot $\tau ''$ with $\tau'\geq \tau''>\tau_0$ where a transition occurred, such that $F_{ij}(\tau'')\geq T$ and $F_{ij}(\tau''-1)< T$ . Then use the facts $\sum_c \phi^c_{ij}(\tau)\geq 0$, $\sum_c \mu^c_{ij}(\tau,\OR)\leq \pacap$ which hold for any $\tau$, and (\ref{eq:queueFij}) to get
\begin{align*}
F_{ij}(\tau''\hspace{-0.08in}-1)&\geq F_{ij}(\tau'')+\hspace{-0.05in}\sum_c \phi^c_{ij}(\tau''\hspace{-0.08in}-1)-\hspace{-0.05in}\sum_c \mu^c_{ij}(\tau''\hspace{-0.08in}-1,\OR)\\
&\geq T+0-\pacap \stackrel{(\ref{eq:thres})}{>} T_0.
\end{align*}
Thus, since $F_{ij}(\tau''\hspace{-0.08in}-1)>T_0$ we may apply  Lemma~\ref{lem:leaky} on slot $\tau''-1$ to conclude that $\sum_c \phi^c_{ij}(\tau''-1)=\pacap$. Then combine with $\sum_c \mu^c_{ij}(\tau,\OR)\leq \pacap$  and (\ref{eq:queueFij}) again
\begin{align*}
F_{ij}(\tau'')&\leq F_{ij}(\tau''\hspace{-0.08in}-1)-\hspace{-0.05in}\sum_c \phi^c_{ij}(\tau''\hspace{-0.08in}-1)+\hspace{-0.05in}\sum_c \mu^c_{ij}(\tau''\hspace{-0.08in}-1,\OR)\\
&< T-\pacap+\pacap = T.
\end{align*}
which is  a contradiction.
\end{proof}
To prove Lemma~\ref{cor:stat}, we will first show that for any router edge $(i,j)$ it is
\begin{equation}\label{eq:ORefficiency}
Kf_{ij}^{c,\boldsymbol\lambda^{\epsilon}}-F^{\max}\leq   \mathbbm{E}_{ \mathbf{H}}\left\{\tilde\mu_{ij}^{c}(t,\OR) \right\} \leq Kf_{ij}^{c,\boldsymbol\lambda^{\epsilon}}
\end{equation}
We begin with  the RHS of (\ref{eq:ORefficiency}). 
For any slot $\tau$ in the observation period $\{t,t+1,\dots,t+K-1\}$, observe that if the value of $F_{ij}(\tau)$ is revealed,  $\mu_{ij}^{c}(\tau,\OR)$ does not depend further on $\mathbf{H}(t)$, 
i.e.,  $\mu_{ij}^{c}(\tau,\OR)$ and $\mathbf{H}(t)$ are conditionally mutually independent
and
 we may write
\begin{align}\label{eq:conditional}
&\mathbbm{E}_{ \mathbf{H}}\left\{\mu_{ij}^{c}(\tau,\OR)|F_{ij}(\tau)<T\right\}\notag\\ & \hspace{0.3in} \triangleq \mathbbm{E}\left\{\mu_{ij}^{c}(\tau,\OR)|F_{ij}(\tau)<T,\mathbf{H}(t)\right\}\notag\\& \hspace{0.3in}=\mathbbm{E}\left\{\mu_{ij}^{c}(\tau,\OR)|F_{ij}(\tau)<T\right\}.
\end{align}
Then, by the law of total expectation we have for $P(F_{ij}(\tau)<T|\mathbf{H}(t))>0$
\begin{align*}\label{eq:RHS}
&\mathbbm{E}_{ \mathbf{H}}\left\{\mu_{ij}^{c}(\tau,\OR)\right\}=\notag\\
&\hspace{0.28in}
=
P(F_{ij}(\tau)<T|\mathbf{H}(t))\mathbbm{E}_{ \mathbf{H}}\left\{\mu_{ij}^{c}(\tau,\OR)|F_{ij}(\tau)<T\right\}\notag\\
&\hspace{0.32in}
+
P(F_{ij}(\tau)\geq T|\mathbf{H}(t))\mathbbm{E}_{ \mathbf{H}}\left\{\mu_{ij}^{c}(\tau,\OR)|F_{ij}(\tau)\geq T\right\}\notag\\
&\hspace{0.24in}\stackrel{(\ref{eq:conditional})}{=}P(F_{ij}(\tau)<T|\mathbf{H}(t))\mathbbm{E}\left\{\mu_{ij}^{c}(\tau,\OR)|F_{ij}(\tau)<T\right\}\notag\\
&\hspace{0.24in}\stackrel{(\ref{eq:allocation})}{\leq} f_{ij}^{c,\boldsymbol\lambda^{\epsilon}},
\end{align*}
where we used $\mathbbm{E}\left\{\mu_{ij}^{c}(\tau,\OR)|F_{ij}(\tau)\geq T\right\}=0$ by definition of $\OR$.
For $P(F_{ij}(\tau)<T|\mathbf{H})=0$ we immediately get $\mathbbm{E}_{ \mathbf{H}}\left\{\mu_{ij}^{c}(\tau,\OR)\right\}= 0\leq f_{ij}^{c,\boldsymbol\lambda^{\epsilon}}$. 
 Summing up over all slots proves the RHS of \eqref{eq:ORefficiency}. 
 
 To prove the LHS of (\ref{eq:ORefficiency}) we will use Lemma \ref{lem:absorption}. First assume that the observation period starts with $F_{ij}(t)<T$. Then invoking Lemma \ref{lem:absorption} we conclude that $F_{ij}(\tau)<T$ for all $\tau=t,\dots,t+K-1$ for any realization of the system evolution. Then assume that  the observation period starts with $F_{ij}(t)>T$, by \eqref{eq:zeroallocation} we have $\sum_c\mu_{ij}^c(t,\OR)=0$ and it follows that the tunnel backlog monotonically decreases until it becomes less than $T$. Moreover, since $F_{ij}(t)<F^{\max}$, the maximum number of slots required to become smaller than $T$ is at most $\left\lceil \frac{F^{\max}-T}{\pacap}\right\rceil$. On the first slot when $F_{ij}(\tau)<T$, we can apply Lemma \ref{lem:absorption} again. Thus, combining the two cases, we conclude that for any realization we have
 \[
 F_{ij}(\tau)<T,\quad \text{for all } \tau=t+\left\lceil \frac{F^{\max}-T}{\pacap}\right\rceil,\dots, t+K-1.
 \]
 Let $\tau_1\triangleq \left\lceil \frac{F^{\max}-T}{\pacap}\right\rceil$, we have
 
 \begin{align*}
  \mathbbm{E}_{ \mathbf{H}}\left\{\tilde\mu_{ij}^{c}(t,\OR) \right\} &\geq \sum_{\tau=t+\tau_1}^{t+K-1}\mathbbm{E}_{ \mathbf{H}}\left\{\mu_{ij}^{c}(\tau,\OR) \right\} \\
  &\hspace{-0.65in}=\sum_{\tau=t+\tau_1}^{t+K-1} \mathbbm{E}_{ \mathbf{H}}\left\{\mu_{ij}^{c}(\tau,\OR)| F_{ij}(\tau)<T \right\}\\
    &\hspace{-0.65in}=(K-\tau_1)f_{ij}^{c,\boldsymbol\lambda^{\epsilon}}> Kf_{ij}^{c,\boldsymbol\lambda^{\epsilon}}-\tau_1\pacap\\
    &\hspace{-0.65in}= Kf_{ij}^{c,\boldsymbol\lambda^{\epsilon}}- \left\lceil \frac{F^{\max}-T}{\pacap}\right\rceil\pacap\\
      &\hspace{-0.65in}\geq Kf_{ij}^{c,\boldsymbol\lambda^{\epsilon}}- (F^{\max}-T+\pacap)> Kf_{ij}^{c,\boldsymbol\lambda^{\epsilon}}-F^{\max}
 \end{align*}
 where the last inequality follows from $T>\pacap$, see \eqref{eq:thres}. 
%
This proves \eqref{eq:ORefficiency}.
To complete the proof, 
we use the lower bound of eq.~(\ref{eq:ORefficiency}) for the first term and the upper bound for the second term, and use the fact that node's $i$ out-degree  is bounded above by the maximum node degree $d_{\max}$. 
\end{proof}

\section{Proof of Lemma~\ref{lem:kslot}}

\newtheorem*{lem:7}{Lemma \ref{lem:kslot}}
\begin{lem:7}[$K$-slot comparison $\TB$ vs $\OR$]
The $K$-slot policy comparison yields for all $(i,j)\in{\cal E}$
\begin{align}\label{eq:nathan}
&\mathbbm{E}_{\mathbf H}\left\{\sum_{c} \tilde\mu_{ij}^{c}(t,\TB)  \left[ Q_i^c(t)-Q_j^c(t)\right]  \right\} \\
&  \hspace{0.1in} \geq \mathbbm{E}_{\mathbf H}\left\{\sum_{c}\tilde \mu_{ij}^{c}(t,\OR)\left[ Q_i^c(t)-Q_j^c(t)\right] \right\} -K^2B_2,\notag
\end{align}
where $B_2$ is a constant given in eq.~(\ref{eq:B1}).
\end{lem:7}

\begin{proof}[Proof of Lemma \ref{lem:kslot}]
Fix some arbitrary router edge $(i,j)$, and  a time slot $t$. 
The concept of the proof is to examine the subsequent $K$ slots and compare $\TB$ to $\OR$ with respect to the products $\mathbbm{E}_{\mathbf H}\left\{\sum_{c}\tilde\mu_{ij}^{c}(t,\pi)  \left[ Q_i^c(t)-Q_j^c(t)\right]\right\}$, where $Q_i^c(t), Q_j^c(t)$ are fully determined by $\mathbf{H}(t)$, and $\tilde\mu_{ij}^{c}(t,\pi) \triangleq \sum_{\tau=0}^{K-1}\mu_{ij}^{c}(t+\tau,\pi)$ represents the decisions made by policy $\pi$ in the $K$-slot observation period starting at time $t$ and state $\mathbf{H}(t)$. 
To avoid a possible confusion, we note that  $Q_i^c(t+\tau),~\tau=0,\dots,K-1$ denote  backlogs  under the $\TB$ policy. 
Although the initial state is common to both policies, the evolution through the $K$-slot period might be different, see for example  Figure \ref{fig:trajectory}. 

We first make a few definitions that regard the sample path evolution of the system under $\TB$  within the observation period of slots ${\cal K}\triangleq \{t,\dots,t+K-1\}$. To make the notation compact, we define a random vector $S:\Omega\to \{0,1\}^K$ such that for any realization $\omega$ and any $t+\tau\in {\cal K}$ it is
\[
S_{\tau}(\omega)=\left\{\begin{array}{ll}
1 &  \text{ if } F_{ij}(t+\tau,\omega)>T\\
0 &  \text{ if } F_{ij}(t+\tau,\omega)\leq T.
\end{array}\right.
\]
Fix a sample path $\omega \in\Omega$. This corresponds to particular vector $S(\omega)$.
If $S_{\tau}=1$ we say that the slot $t+\tau$ \emph{is  overload}. 
Let ${\cal O}\subseteq {\cal K}$ be the set of all overload slots. Similarly if $S_{\tau}=0$, we say that the slot $t+\tau$ \emph{is underload} and denote the corresponding set with ${\cal U}={\cal K-O}$. 
We remark that these sets are realized for the specific sample path.
In the following, we will compare $\TB$ to $\OR$ for this sample path.

First we compare the two policies across underload slots, $t+\tau\in {\cal U}$. In such slots we have by $\TB$ design that
\begin{align}\label{eq:globalbnd}
&\sum_{c}\mu_{ij}^{c}(t+\tau,\text{\TB})  \left[ Q_i^c(t+\tau)-Q_j^c(t+\tau)\right]\geq \notag\\
&\hspace{0.4in} \sum_{c}\mu_{ij}^{c}(t+\tau,\OR)  \left[ Q_i^c(t+\tau)-Q_j^c(t+\tau)\right]\\
&\hspace{2.2in},~\forall t+\tau\in{\cal U}\notag
\end{align}
where we emphasize that $\mu_{ij}^{c}(t+\tau,\OR)$ is not decided based on $Q_i^c(t+\tau),Q_j^c(t+\tau)$.\footnote{This is because $Q_i^c(t+\tau),Q_j^c(t+\tau)$  are the backlogs at $t+\tau$ under $\TB$, but not necessarily under $\OR$.} Nevertheless the inequality holds since, given underload, $\TB$ is a universal maximizer for this quantity. 

We will need a bound for the largest backlog increase and decrease in $k$ slots. 
Let  $\delta Q_i^c(k) \triangleq  Q_i^c(t+k)-Q_i^c(t)$, we have
\begin{equation*}\label{eq:kslot}
-k \sum_{b\in\text{Out($i$)}} R_{ib} \leq \delta Q_i^c(k)\leq k (\sum_{a\in \text{In($i$)}} R_{ai}+A_{\max}).
\end{equation*}
which are independent of $t$. Also recall that 
 $R_{\max}$ is the maximum link capacity and $d_{\max}$ the maximum node degree on ${\cal G}_R$, and define
 \begin{equation}\label{eq:B1}
 B_2\triangleq R_{\max} \left(2d_{\max}R_{\max}+A_{\max}\right).\end{equation}
 It follows that $-kB_2\leq R_{\max} \left[\delta Q_i^c(k)-\delta Q_j^c(k)\right]\leq kB_2$.
 Also, note that under any policy $\pi$ it is $\sum_c \mu_{ij}^c(t+\tau,\pi)\leq R_{\max}$.
Then, on an underload slot  $t+\tau\in {\cal U}$, we have
\begin{align}\label{eq:underloadslots}
&\sum_{c}\mu_{ij}^{c}(t+\tau,\text{\TB})  \left[ Q_i^c(t)-Q_j^c(t)\right] \notag\\
&\hspace{0.1in}\geq \sum_{c}\mu_{ij}^{c}(t+\tau,\text{\TB})  \left[ Q_i^c(t+\tau)-Q_j^c(t+\tau)\right]-\tau B_2 \notag\\
&\hspace{0.08in} \stackrel{\text{(\ref{eq:globalbnd})}}{\geq}\sum_{c}\mu_{ij}^{c}(t+\tau,\OR)  \left[ Q_i^c(t+\tau)-Q_j^c(t+\tau)\right]-\tau B_2\notag \\
&\hspace{0.1in}= \sum_{c}\mu_{ij}^{c}(t+\tau,\OR)  [ Q_i^c(t)-Q_j^c(t)+\notag \\
&\hspace{1.7in}+\delta Q_i^c(\tau)-\delta Q_j^c(\tau)] -\tau B_2 \notag \\
&\hspace{0.1in}\geq \sum_{c}\mu_{ij}^{c}(t+\tau,\OR)  \left[ Q_i^c(t)-Q_j^c(t)\right]-2\tau B_2
\end{align}

\begin{figure}[t!]
\begin{center}
\begin{overpic}[scale=.35]{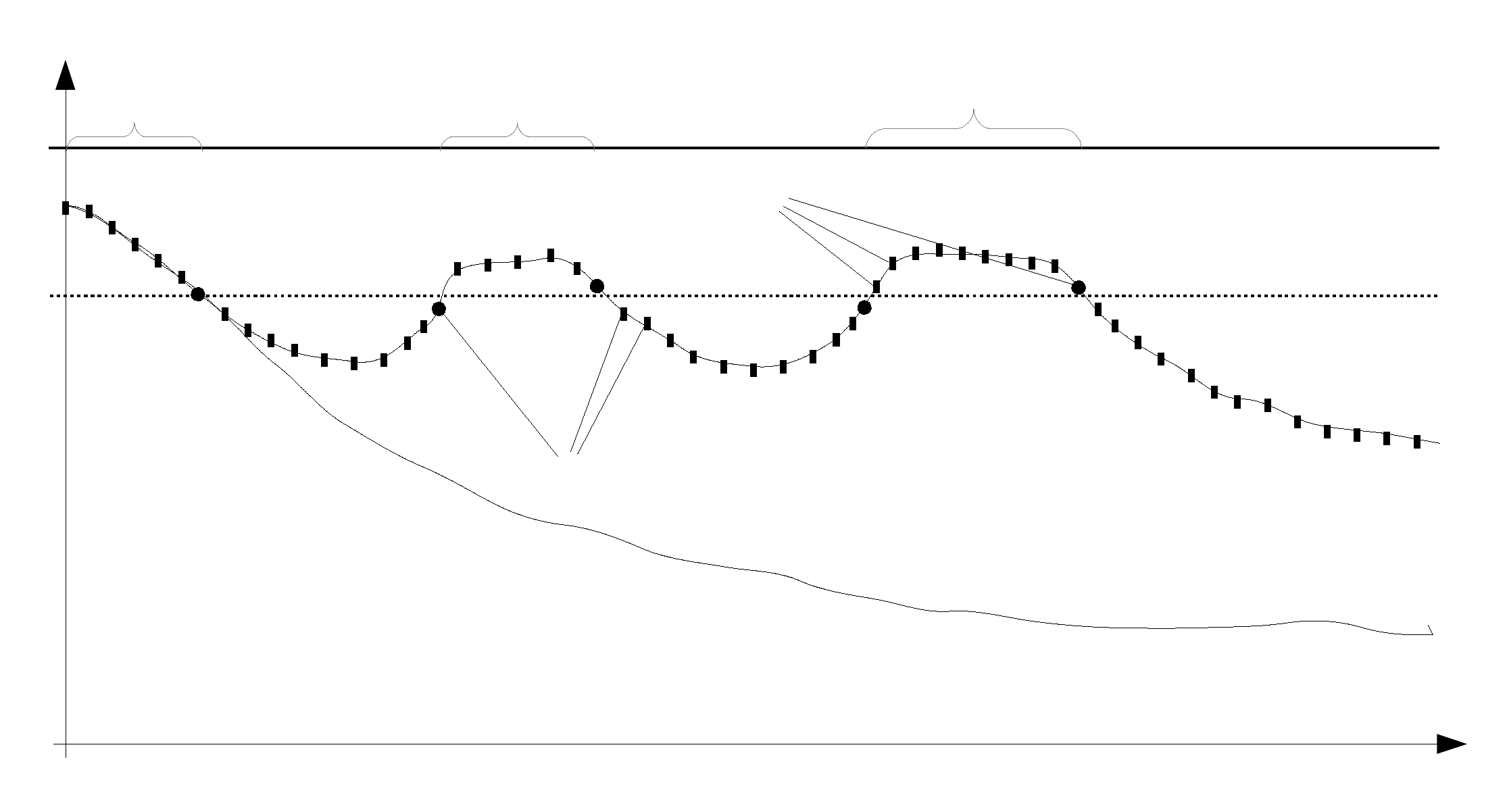}
\put(-7,47){\footnotesize $F_{ij}(\tau)$}
\put(4,-2){\footnotesize $t$}
\put(86,-2){\footnotesize $t+K-1$}
\put(34.5,21){\footnotesize $\in {\cal U}$}
\put(45,40){\footnotesize $\in {\cal O}$}
\put(7,48){\footnotesize ${\cal T}_1$}
\put(32,48){\footnotesize ${\cal T}_2$}
\put(63,48){\footnotesize ${\cal T}_3$}
\put(95,35){\footnotesize $T$}
\put(95,45){\footnotesize $F^{\max}$}
\put(83.5,28.3){\colorbox{white}{\textcolor{black}{\footnotesize $\TB$}}}
\put(82,15){\colorbox{white}{\textcolor{black}{\footnotesize $\OR$}}}
\end{overpic}
   \caption{Sample path comparison of the two policies over $K$ slots starting from the same state. We note an overload subperiod starts at a slot where $F_{ij}(\tau)<T$ (with the possible exemption of the first overload subperiod) and ends at a slot where $F_{ij}(\tau)>T$.}
      \vspace{-0.4in}
     \label{fig:trajectory}
\end{center}
\end{figure}
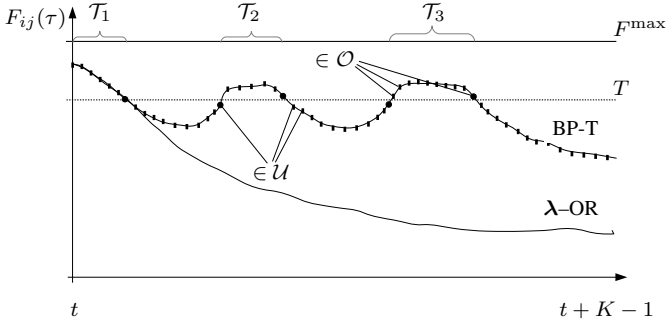

A similar  bound is derived previously in \cite{NMR05} to be applied to a $K$-slot comparison where the stationary policy does not depend on the backlog sizes.

Our plan is to derive a similar expression to \eqref{eq:underloadslots} for the overload slots.
To proceed with the plan, we develop an analysis which depends on  the  sign of $\left[ Q_i^c(t)-Q_j^c(t)\right]$ which is determined at the beginning of the $K$-slot period. If positive, we break the observation into overload subperiods ${\cal T}$ (to be defined shortly) and the remaining underload slots ${\cal K-T}$. If negative, then we study separately the overload slots ${\cal O}$ and the remaining underload slots ${\cal K-O}$. 

a) Assume first that the observed state $\mathbf{H}(t)$ is such that $\left[ Q_i^c(t)-Q_j^c(t)\right]\geq 0$. 
For this case, we use the concept of an \emph{overload subperiod}, which is a period of consecutive overload slots plus an initial underload slot.

We formally define the $m^{\text{th}}$ {overload subperiod} with length $L_m$  consisting of consecutive slots $\{t+\tau_1^m,\dots,t+\tau_{L_m}^m\}$, such that $S_{\tau_1^m}=S_{\tau_{L_m}^m+1}=0$ and $S_{\tau}=1,~\forall \tau\in\{\tau_2^m,\dots,\tau_{L_m}^m\}$. In words, an overload subperiod begins with one underload slot and ends with an overload slot, while all slots within the subperiod are overload and the slot after the subperiod is underload, see a representation of such an overload subperiod in Fig.~\ref{fig:trajectory}. Let ${\cal T}_m$ be the set of slots  comprising the $m^{\text{th}}$ overload subperiod for sample path under study. Suppose, that there are $Z(\omega)$ overload subperiods, where the random variable $Z$ takes values in $\{0,1,\dots,\lceil K/2 \rceil\}$. 
We also define ${\cal T}=\cup_{m=1}^{Z}{\cal T}_m$. Note that the sets ${\cal T}_m$ are disjoint, it is ${\cal T}\subseteq {\cal K}$, and ${\cal K}-{\cal T}\subseteq {\cal U}$.

By definition of the overload subperiod the backlog at the last slot is larger than at the first slot, hence for our chosen sample path we have
\begin{equation}\label{eq:bound1}
F_{ij}(t+\tau_{L_m}^m)-F_{ij}(t+\tau_1^m)>0,~~ \text{ for } m=2,3,\dots,Z.
\end{equation}
Let us now extend the definition of the overload subperiod to the  special case of the first subperiod. If the first slot of the observation period is overload, i.e., $S_{0}=1$, then the first overload subperiod starts at an overload slot (as opposed to the original definition) and completes at the last consecutive overload slot (similar to the original definition).\footnote{Similarly, if the last subperiod ends at an overload slot, then we do not have a followup underload slot-however this case does not affect our proof.}
This is a natural extension to the above definition of the overload subperiod. 
The backlog difference between last and first slot of the first overload subperiod is
\begin{align}
&F_{ij}(t+\tau_{L_1}^1)-F_{ij}(t)>0~~~~~\quad\quad \text{ if } S_0=0\label{eq:bound21}\\
&F_{ij}(t+\tau_{L_1}^1)-F_{ij}(t)>-\bound~~~~ \text{ if } S_0=1.\label{eq:bound2}
\end{align}

Now, let us examine the $m^{\text{th}}$ overload subperiod  of slots ${\cal T}_m$ for $m>1$, combining  (\ref{eq:bound1}) and (\ref{eq:kFeq2})  we have
\begin{align*}
\sum_{c,t+\tau\in {\cal T}_m}\hspace{-0.15in}\mu_{ij}^{c}(t+\tau,\text{\TB})&\geq \sum_{c,t+\tau\in {\cal T}_m}  \phi_{ij}^{c}(t+\tau)\\
&=| {\cal T}_m| \pacap \\
&\stackrel{\eqref{eq:allocation}}{\geq} \sum_{c,t+\tau\in {\cal T}_m}\mu_{ij}^{c}(t+\tau,\OR),
\end{align*}
where the equality follows from applying Lemma \ref{lem:leaky} to all slots in the overload subperiod (including the first).
Multiplying both sides with the positive quantity $\left[ Q_i^c(t)-Q_j^c(t)\right]$, we get for overload periods $m>1$ starting from a state with positive $\left[ Q_i^c(t)-Q_j^c(t)\right]$

\begin{align}\label{eq:overloadpos}
&\sum_{c,t+\tau\in {\cal T}_m}\mu_{ij}^{c}(t+\tau,\text{\TB})  \left[ Q_i^c(t)-Q_j^c(t)\right] \\
&\hspace{0.1in}\geq \hspace{-0.25in}\sum_{\hspace{0.25in}c,t+\tau\in {\cal T}_m}\hspace{-0.25in}\mu_{ij}^{c}(t+\tau,\OR)  \left[ Q_i^c(t)-Q_j^c(t)\right]\notag\\
&\hspace{0.1in}\geq \hspace{-0.2in}\sum_{\hspace{0.2in}c,t+\tau\in {\cal T}_m}\hspace{-0.2in}\mu_{ij}^{c}(t+\tau,\OR)  \left[ Q_i^c(t)-Q_j^c(t)\right]- \hspace{-0.25in}\sum_{\hspace{0.15in}t+\tau\in {\cal T}_m}\hspace{-0.15in}2\tau B_2\notag
\end{align}
where in the last step we intentionally relaxed the bound further to make it match (\ref{eq:underloadslots}).
For $m=1$ and $S_0=0$, we repeat the above approach using (\ref{eq:bound21}), and  (\ref{eq:overloadpos})  still holds. However, in case  $S_0=1$, i.e. the observation period starts in overload,  we  must replace  (\ref{eq:bound1}) with (\ref{eq:bound2}),  in which case the above approach breaks. Therefore we deal with this case in a different manner. In particular we will show that if our sample path has  $S_0=1$ then  for all time slots in the first overload subperiod $t+\tau\in{\cal T}_1$,  
\[
\sum_c\mu_{ij}^{c}(t+\tau,\text{\TB})=\sum_c\mu_{ij}^{c}(t+\tau,\OR)=0.
\]
Starting from the first slot $t$, and since $S_0=1\Leftrightarrow F_{ij}(t)>T$, observe that both policies $\TB,\OR$ will make the same decision $\sum_c\mu_{ij}^{c}(t,\pi)=0$. Then  (\ref{eq:kFeq2}) is satisfied with equality, and since  $\phi_{ij}^c(t)$ does not depend on the chosen policy, we have that $F_{ij}(t+1)$ is the same for both policies. This process is repeated for all slots in subperiod ${\cal T}_1$ consisting of  overload slots under $\TB$. Thus, we conclude that if the system is in the first overload period under $\TB$ with $S_0=1$, then it is also in the first overload period under $\OR$.
Therefore, for $t+\tau\in{\cal T}_1$, $S_0=1$ we have
\[
\sum_{c,t+\tau\in {\cal T}_1}\mu_{ij}^{c}(t+\tau,\text{\TB})=\sum_{c,t+\tau\in {\cal T}_1}\mu_{ij}^{c}(t+\tau,\OR)
\]
and (\ref{eq:overloadpos}) holds for this case as well. 
We conclude that (\ref{eq:overloadpos}) is true for all $m$ as long as $\left[ Q_i^c(t)-Q_j^c(t)\right]\geq 0$. 

Let $\mathbf{Q}^+_t$ denote the event $\left[ Q_i^c(t)-Q_j^c(t)\right]\geq 0$ and  $\mathbf{Q}^-_t$ the complement.
Observing that the remaining slots are underload ${\cal K-T}\subseteq {\cal U}$ and combining  with ineq. (\ref{eq:underloadslots}), we condition on the sample path $S=\bm s$ to get
\begin{align}\label{eq:Qpos}
& \mathbbm{E}_{\mathbf H}\left\{\sum_{c} \tilde\mu_{ij}^{c}(t,\TB)  \left[ Q_i^c(t)-Q_j^c(t)\right]\big |\mathbf{Q}^+_t,S=\bm s\right\} \notag\\
& =\mathbbm{E}_{\mathbf H}\left\{\hspace{-0.2in}\sum_{\hspace{0.2in}c,t+\tau\in {\cal T}}\hspace{-0.28in} \mu_{ij}^{c}(t+\tau,\TB)  \left[ Q_i^c(t)-Q_j^c(t)\right]\big |\mathbf{Q}^+_t,S=\bm s\right\} \notag\\
& +\mathbbm{E}_{\mathbf H}\left\{\hspace{-0.3in}\sum_{\hspace{0.25in}c,t+\tau\in {\cal K-T}} \hspace{-0.37in}\mu_{ij}^{c}(t+\tau,\TB)  \left[ Q_i^c(t)-Q_j^c(t)\right]\big |\mathbf{Q}^+_t,S=\bm s\right\} \notag\\
&\hspace{-0.05in}\stackrel{\text{(\ref{eq:underloadslots})\&(\ref{eq:overloadpos})}}{\geq} \mathbbm{E}_{\mathbf H}\left\{\sum_{c} \tilde\mu_{ij}^{c}(t,\OR)  \left[ Q_i^c(t)-Q_j^c(t)\right]\big |\mathbf{Q}^+_t,S=\bm s\right\} \notag\\
&-\sum_{t+\tau\in {\cal K}}2\tau B_2
\end{align}

b) Next we study the case where the observation period starts with $\left[ Q_i^c(t)-Q_j^c(t)\right]< 0$ and we examine the  overload slots. Since $\TB$ refrains from transmission in these slots, we have 
\[
\sum_c\mu_{ij}^{c}(t+\tau,\text{\TB})=0\leq \sum_c\mu_{ij}^{c}(t+\tau,\OR), ~\forall t+\tau\in {\cal O},
\]
multiplying with the negative quantity $\left[ Q_i^c(t)-Q_j^c(t)\right]$ we get 
\begin{align}\label{eq:overloadneg}
&\sum_{c,t+\tau\in  {\cal O}}\mu_{ij}^{c}(t+\tau,\text{\TB})  \left[ Q_i^c(t)-Q_j^c(t)\right] \notag\\
&\hspace{0.1in}\geq \hspace{-0.1in}\sum_{c,t+\tau\in  {\cal O}}\hspace{-0.1in}\mu_{ij}^{c}(t+\tau,\OR)  \left[ Q_i^c(t)-Q_j^c(t)\right]\notag\\
&\hspace{0.1in}>\hspace{-0.1in} \sum_{c,t+\tau\in  {\cal O}}\hspace{-0.1in}\mu_{ij}^{c}(t+\tau,\OR)  \left[ Q_i^c(t)-Q_j^c(t)\right]-2\tau B_2.
\end{align}
\noindent Combining with (\ref{eq:underloadslots}) we obtain
\allowdisplaybreaks
\begin{align}\label{eq:Qneg}
&\mathbbm{E}_{\mathbf H}\left\{\sum_{c} \tilde\mu_{ij}^{c}(t,\TB)  \left[ Q_i^c(t)-Q_j^c(t)\right]\big |\mathbf{Q}^-_t,S=\bm s\right\} \notag\\
& =\mathbbm{E}_{\mathbf H}\left\{\hspace{-0.2in}\sum_{\hspace{0.15in}c,t+\tau\in {\cal O}}\hspace{-0.22in} \mu_{ij}^{c}(t+\tau,\TB)  \left[ Q_i^c(t)-Q_j^c(t)\right]\big |\mathbf{Q}^-_t,S=\bm s\right\} \notag\\
& +\mathbbm{E}_{\mathbf H}\left\{\hspace{-0.3in}\sum_{\hspace{0.25in}c,t+\tau\in {\cal K-O}} \hspace{-0.35in}\mu_{ij}^{c}(t+\tau,\TB)  \left[ Q_i^c(t)-Q_j^c(t)\right]\big |\mathbf{Q}^-_t,S=\bm s\right\} \notag\\
&\hspace{-0.01in}\stackrel{\text{(\ref{eq:underloadslots})\&(\ref{eq:overloadneg})}}{\geq} \mathbbm{E}_{\mathbf H}\left\{\sum_{c} \tilde\mu_{ij}^{c}(t,\OR)  \left[ Q_i^c(t)-Q_j^c(t)\right]\big |\mathbf{Q}^-_t,S=\bm s\right\} \notag\\
&-\sum_{t+\tau\in {\cal K}}2\tau B_2
\end{align}
In conclusion, depending on the  sign of $\left[ Q_i^c(t)-Q_j^c(t)\right]$, we either break the observation into overload subperiods ${\cal T}$ and  remaining underload slots ${\cal K-T}$ to use (\ref{eq:overloadpos}) and  (\ref{eq:underloadslots}), or  we study separately the overload slots ${\cal O}$ and the remaining underload slots ${\cal K-O}$ using (\ref{eq:overloadneg}) and (\ref{eq:underloadslots}). 
Note that $K^2>K(K-1)\triangleq \sum_{\tau=0}^{K-1} 2\tau=\sum_{t+\tau\in {\cal K}}2\tau$.
Hence
\begin{align}\label{eq:allcases1}
&\hspace{-0.01in} \mathbbm{E}_{\mathbf H}\left\{\sum_{c} \tilde\mu_{ij}^{c}(t,\TB)  \left[ Q_i^c(t)-Q_j^c(t)\right]\big |S=\bm s\right\} \notag\\
& =\mathbbm{E}_{\mathbf H}\left\{\sum_{c} \tilde\mu_{ij}^{c}(t,\TB)  \left[ Q_i^c(t)-Q_j^c(t)\right]\big |\mathbf{Q}^+_t,S=\bm s\right\} \notag\\
& +\mathbbm{E}_{\mathbf H}\left\{\sum_{c} \tilde\mu_{ij}^{c}(t,\TB)  \left[ Q_i^c(t)-Q_j^c(t)\right]\big |\mathbf{Q}^-_t,S=\bm s\right\} \notag\\
&\hspace{-0.01in}\stackrel{\text{(\ref{eq:Qpos})\&(\ref{eq:Qneg})}}{>}  \hspace{-0.01in}\mathbbm{E}_{\mathbf H}\left\{\sum_{c} \tilde\mu_{ij}^{c}(t,\OR)  \left[ Q_i^c(t)-Q_j^c(t)\right]\big |S=\bm s\right\}\notag\\
& -K^2 B_2.
\end{align}
Let ${\cal S}=\{\bm s:\mathbf{H}(t)\cap (S=\bm s) \neq \emptyset\}$, we have
\begin{align*}
&\hspace{-0.03in} \mathbbm{E}_{\mathbf H}\left\{\sum_{c} \tilde\mu_{ij}^{c}(t,\TB)  \left[ Q_i^c(t)-Q_j^c(t)\right]\right\} \notag\\
&=\sum_{\bm s\in {\cal S}}P(S=\bm s|\mathbf{H}(t))\\
&\hspace{0.4in} \times \mathbbm{E}_{\mathbf H}\left\{\sum_{c} \tilde\mu_{ij}^{c}(t,\TB)  \left[ Q_i^c(t)-Q_j^c(t)\right]\big |S=\bm s\right\} \notag\vspace{-0.2in}\\
&\hspace{-0.05in}\stackrel{(\ref{eq:allcases1})}{\geq} \sum_{\bm s\in {\cal S}}P(S=\bm s|\mathbf{H}(t))\vspace{-0.2in}\\
&\hspace{0.4in} \times\mathbbm{E}_{\mathbf H}\left\{\sum_{c} \tilde\mu_{ij}^{c}(t,\OR)  \left[ Q_i^c(t)-Q_j^c(t)\right]\big |S=\bm s\right\} \notag\\
&\hspace{0.4in}- \sum_{\bm s\in {\cal S}}P(S=\bm s|\mathbf{H}(t))K^2B_2\notag\\
&=\mathbbm{E}_{\mathbf H}\left\{\sum_{c} \tilde\mu_{ij}^{c}(t,\OR)  \left[ Q_i^c(t)-Q_j^c(t)\right]\right\} -K^2B_2.
\end{align*}
\end{proof}

\end{document}